\documentclass[11pt,a4paper]{article}
\usepackage[lmargin=1.0in,rmargin=1.0in,bottom=1.0in,top=1.0in,twoside=False]{geometry}

\usepackage{inconsolata}
\usepackage{libertine}

\usepackage{amssymb,amsmath,amsthm}
\usepackage{mathtools}
\usepackage{graphicx}%
\usepackage{enumerate}
\usepackage{tikz}
\usetikzlibrary{shapes}

\usepackage[T1]{fontenc}

\usepackage{enumitem}

\usepackage{microtype}
\usepackage{amsfonts}
\usepackage{comment}
\usepackage{mathrsfs}
\usepackage{todonotes}

\definecolor{blue}{rgb}{0.1,0.2,0.5}
\definecolor{brown}{rgb}{0.6,0.6,0.2}
\usepackage[ocgcolorlinks, linkcolor={blue}, citecolor={brown}]{hyperref}
\usepackage{cleveref}

\usepackage{enumerate}
\usepackage{latexsym}

\newtheorem{lemma}{Lemma}[section]
\newtheorem{claim}{Claim}[section]
\newtheorem{theorem}[lemma]{Theorem}
\newtheorem{corollary}[lemma]{Corollary}
\newtheorem{question}{Question}

\newcommand{\partition}{\mathcal{C}}
\newcommand{\Oh}{\mathcal{O}}
\newcommand{\eps}{\varepsilon}

\renewcommand{\leq}{\leqslant}
\renewcommand{\geq}{\geqslant}
\renewcommand{\le}{\leqslant}

\renewcommand{\setminus}{-}

\newcommand{\Ll}{\mathcal{L}}
\newcommand{\wh}[1]{\widehat{#1}}
\newcommand{\Aa}{\mathcal{A}}
\newcommand{\Dd}{\mathcal{D}}
\newcommand{\Ff}{\mathcal{F}}

\newcommand{\polylog}{\mathrm{polylog}}
\newcommand{\poly}{\mathrm{poly}}
\newcommand{\Prob}{\mathbb{P}}
\newcommand{\Exp}{\mathbb{E}}
\newcommand{\ExpDist}{\mathrm{Exp}}

\newcommand{\len}{\mathrm{length}}

\newcommand{\R}{\mathbb{R}}
\newcommand{\dd}{\mathrm{d}}

\newcommand{\dist}{\mathrm{dist}}

\begin{document}

\author{
\and
Vincent Cohen-Addad\thanks{Google Research, France (\texttt{cohenaddad@google.com})}
\and
Hung Le\thanks{University of Massachusetts Amherst, USA (\texttt{hungle@cs.umass.edu})}
\and
Marcin Pilipczuk\thanks{Institute of Informatics, University of Warsaw, Poland and IT University of Copenhagen, Denmark (\texttt{marcin.pilipczuk@mimuw.edu.pl})}
\and
Michał Pilipczuk\thanks{Institute of Informatics, University of Warsaw, Poland (\texttt{michal.pilipczuk@mimuw.edu.pl})}
}
\title{Planar and Minor-Free Metrics Embed into Metrics of Polylogarithmic Treewidth with Expected Multiplicative Distortion Arbitrarily Close to 1
  \thanks{
 This work is 
a part of projects CUTACOMBS (Ma. P.) and BOBR (Mi. P.) that have received funding from the European Research Council (ERC) 
under the European Union's Horizon 2020 research and innovation programme (grant agreements No.~714704 and~948057, respectively).
Hung Le is supported by the NSF CAREER Award No. CCF-2237288 and an NSF Grant No. CCF-2121952.
Ma. P. is also part of BARC, supported by the VILLUM Foundation grant 16582.}
}

\date{}

\maketitle

\begin{abstract}
We prove that there is a randomized polynomial-time algorithm that given an edge-weighted graph $G$ excluding a fixed-minor $Q$ on $n$ vertices and an accuracy parameter $\eps>0$, constructs an edge-weighted graph~$H$ and an embedding $\eta\colon V(G)\to V(H)$ with the following properties:
\begin{itemize}
 \item For any constant size $Q$, the treewidth of $H$ is polynomial in $\eps^{-1}$, $\log n$, and the logarithm of the stretch of the distance metric in $G$.
 \item The expected multiplicative distortion is $(1+\eps)$: for every pair of vertices $u,v$ of $G$, we have $\dist_H(\eta(u),\eta(v))\geq \dist_G(u,v)$ always and $\Exp[\dist_H(\eta(u),\eta(v))]\leq (1+\eps)\dist_G(u,v)$.
\end{itemize}
Our embedding is the first to achieve polylogarithmic treewidth of the host graph and comes close to the lower bound by Carroll and Goel, who showed that any embedding of a planar graph with $\Oh(1)$ expected
distortion requires the host graph to have treewidth $\Omega(\log n)$. 
It also provides a unified framework for obtaining randomized quasi-polynomial-time approximation schemes for a variety of problems including network design, clustering or routing problems, in minor-free metrics where the optimization goal is the sum of selected distances. Applications include the capacitated vehicle routing problem, and capacitated clustering problems.
\end{abstract}

\section{Introduction}

Tree metrics are among the easiest metrics that an algorithm designer could have to deal with:   a vast number of NP-hard problems in tree metrics can be solved in polynomial time, some even in linear time. Unfortunately, most real-world problems take place in much more complicated metrics.
Hence, the question of how well one can embed a complicated metric into a simpler one, such as for example a tree metric, has naturally emerged in the 90s.
A sequence of works~\cite{Bartal96,Bartal98,FRT04} culminated in showing that any $n$-point metric stochastically embeds into a tree metric such that for
any pair of points, their distance is in expectation preserved up to a multiplicative factor, the \emph{distortion}, of $\Oh(\log n)$; and that such an embedding can be computed
in polynomial time.
This result has a number of algorithmic applications and became a part of the standard toolbox for algorithm design.

Bartal~\cite{Bartal98} showed that the $\Oh(\log n)$ bound on the distortion is the best one can hope for in the worst-case, for arbitrary metric spaces.
To improve upon this bound, one would need to focus on input metrics that exhibit special structure, such as for example shortest-paths metrics in planar and minor-free graphs.
Planar and minor-free metrics are well-studied and have appeared in many practical applications, in particular network design, routing or clustering applications.
In the early 2000s, researchers have thus asked: Could we improve the distortion when the input metric is planar (and the host metrics are trees)?
Caroll and Goel~\cite{CG04} showed that unfortunately, the $\Omega(\log n)$ lower bound also holds when the input metric is planar.

While tree metrics are topologically simple, the ultimate goal of algorithmic designers is to solve algorithmic problems efficiently. In this respect, metrics
induced by graphs of bounded treewith are almost equally powerful while permitting much more complicated structures. Specifically, one can solve an extremely broad
class of problems on graphs of bounded treewidth using dynamic programming on tree decompositions; see~\cite[Chapter~7]{platypus}.
The next question in line was thus whether one can improve over the $\Oh(\log n)$ distortion bound for embedding planar metrics into trees if we instead aim
to embed planar
metrics into bounded treewidth graphs, hence bypassing the lower bound of~\cite{CG04}.

\begin{question}\label{quest-main} Could we embed planar graphs into low-treewidth graphs with small distortion, ideally constant or $(1+\eps)$ for any fixed $\eps \in (0,1)$?
\end{question}

By low treewidth, we mean a constant or even polylogarithmic treewidth -- which would be enough to obtain quasi-polynomial approximation schemes for a large number of NP-hard problems. As we will point out later, constant treewidth is impossible if one insists on $o(\log n)$ distortion.

However, research on \Cref{quest-main} has been riddled with negative results. In particular, any deteministic embedding of planar graphs with  a worst-case constant distortion must have treewidth $\Omega(\sqrt{n})$ of the host graph, which is a nearly trivial bound, since every $n$-vertex planar graph has treewidth $\Oh(\sqrt{n})$~\cite{CG04}. Any randomized embedding of planar graphs into \emph{constant treewidth graphs} must  have expected distortion $\Omega(\log n)$~\cite{CG04,CJVL08}. And any randomized embedding that has  an expected \emph{constant distortion $c$} for \emph{any given constant $c \geq 1$} must have treewidth $\Omega(\log n/c)$~\cite{CG04}.  Given these lower bounds, it seems hopeless to achieve any significant progress on \Cref{quest-main}. On a positive note, some progress has been made recently on embeddings with \emph{additive distortion}, where the input metrics were further restricted to have stretch $2^{\ell}$ for some $\ell > 0$\footnote{The {\em{stretch}} is the ratio between
the maximum pairwise distance and the minimum pairwise distance in a metric.}, and the distortion is additive $+\eps 2^{\ell}$~\cite{EKM13,FKS19,CAFKL20,FL22}. Additive distortion is a much weaker guarantee than multiplicative distortion and is not truly comparable to the $\Oh(\log n)$ multiplicative distortion obtained for general metrics\footnote{If one sets $\eps \approx 2^{-\ell}\log n$, one gets multiplicative distortion $\Oh(\log n)$,
but then the treewidth is proportional to $2^{\ell}/\log n$,
which is way too large in most cases. We are rather concerned with graphs of treewidth polynomial in $\ell$.}. These weaknesses severely limit the kind of algorithmic applications of embeddings with additive distortion.

On the other hand, the aforementioned lower bounds still leave room for  a  slim hope: they do not rule out an embedding with expected distortion $(1+\eps)$ and treewidth $\Oh_{\eps}(\log n)$. The lower bound of treewidth $\Omega(\log n/c)$  mentioned earlier, while holding for any constant distortion $c$, does not exclude the possibility of an embedding with the treewidth of host graph of $\Oh(\log n/\eps)$ for the expected distortion~$(1+\eps)$. Nevertheless, it has remained unclear how to obtain an embedding with any constant distortion,  say, of an $n$-vertex planar graphs into graphs with treewidth significantly better than the trivial bound $\Oh(\sqrt{n})$.  %
We note that the discussion so far applies equally to weighted and unweighted planar graphs; in other words, having the same weight on every edge does not seem  to make the problem easier . The lower bound of Caroll and Goel~\cite{CG04} holds for the unweighted planar grid.

In this work, we show that an embedding with polylogarithmic treewidth is possible for an expected  distortion $(1+\eps)$. Our result is the first positive  progress on \Cref{quest-main}.

 \begin{theorem}\label{thm:main}
 	For every fixed graph $K$ there exists a randomized polynomial-time algorithm
 	that, given an edge-weighted $K$-minor-free graph $G$ and an accuracy parameter $\eps>0$,
 	constructs a probabilistic metric embedding of $G$ with expected distortion $(1+\eps)$
 	into a graph of treedepth\footnote{The {\em{treedepth}} of a graph $G$ is the minimum possible depth (number of vertices on a root-to-leaf path) of an
elimination forest of $G$. The treedepth of a graph is an upper bound on its treewidth.}
 	\[ \Oh_K\left((\ell+\ln n)^6/\eps\cdot \ln^2 n\cdot (\ln n+\ln \ell+\ln (1/\eps))^5\right), \] 
 	where $n$ is the vertex count of $G$ and $\ell$ is the logarithm of the stretch of the metric induced by $G$.
 \end{theorem}

In particular, for planar graphs --- which exclude $K_5$ as a minor --- with edge weights in $[1,n^c]$ for any constant $c$, our treewidth bound in  \Cref{thm:main} is $\polylog(n)$, as $\ell = \log(n^c) = \Oh(\log n)$. In this work, we do not attempt to minimize the degree of the polylogarithmic factor bounding the treewidth. 

\paragraph*{Applications.}
\Cref{thm:main} is a part of a broader research agenda, aiming to   design approximation schemes for optimization problems in structured  metrics. Notable examples of structured metrics include Euclidean and doubling metrics. 
For Euclidean metrics, tremendous progress was immediately made and quasi-polynomial-time approximation schemes
(QPTASs) and PTASs for other geometric problems (including $k$-median) were presented in a series of
subsequent works~\cite{AroraRR98,Arora97}. The work of Arora and Mitchell, which uses the notion of
quad-tree decompositions to break the Euclidean metric instances into subinstances of the problem with limited
interactions, quickly became a framework that solved a variety of problems. Talwar~\cite{Tal04} took one step further by showing that any doubling metric (a generalization of Euclidean metric that does not necessitate an embedding
of the points) embeds with $(1+\eps)$ expected distortion into a graph of treewidth at most $\Oh_{d,\eps}(\ell)$, where $\ell$ is the logarithm of the
stretch.
This immediately led to (1) QPTASs for a variety of problems for which the objective is the sum of the lengths of a
multiset of edges, and for which a QPTAS is known for bounded treewidth graphs; and (2) a QPTAS for TSP beyond Euclidean
spaces.

Unfortunately, the picture for shortest-path metrics of graphs such as metrics induced by
the shortest paths of planar graphs or graphs excluding a fixed minor
proved more challenging to address. Indeed, obtaining PTASes, or even QPTASes for TSP, or other network design problems
such as Steiner tree or $k$-MST, or clustering problems has been an important challenge that has been partially resolved
during the last 25 years. An illustrating example is that no PTAS for the $k$-MST problem was known until
last year~\cite{CA22} and no QPTAS is known for the vehicle routing problem (unless the capacities are assumed to be
absolute constant~\cite{CAFKL20,FKS19}), nor for non-uniform facility location in minor-free metrics.
The non-uniform facility location problem is indeed an illustrative example: The question of generalizing results for planar graphs
to more general families of graphs, such as graph excluding a fixed minor,
has been an important research agenda. While Grigni~\cite{Grigni00} gave a QPTAS for TSP for minor-free graphs, it was only recently improved
to EPTAS by Borradaile et al.~\cite{BLW17}. These were recently improved (and extended to the subset
version\footnote{In the subset version, the problem is to visit  given subset of the input vertices.})
by Le~\cite{Le20} and Cohen-Addad et al.~\cite{CAFKL20}. The recent result of Cohen-Addad~\cite{CA22} is a step toward
a framework that generalizes planar results on several network design problems (e.g. Steiner forest) to minor-free metrics.
Unfortunately, the approach fails for more sophisticated objectives such as vehicle routing or clustering (facility location or
$k$-median).

The gap between planar (and more generally minor-free) metrics on the one hand and Euclidean and doubling metrics on the other hand is
mainly due to the embedding to treewidth-$\Oh_{d,\eps}(\ell)$ graphs of Talwar~\cite{Tal04} which provided QPTASs for most metric optimization problems
(except for $k$-means or $k$-center where the objective is to minimize a sum of squared distances or a maximum distance and for which
different techniques were used~\cite{CAKM19,CAFS21,FriggstadRS19,FKS19}). %
Our \Cref{thm:main} overcomes the fundamental embedding bottleneck. As corollaries, we obtain QPTASes for  several problems  where  prior techniques were failing.

The first problem is the Capacitated Vehicle Routing Problem (CVRP) in planar metrics. In this problem, we are given: (i) a set of clients represented by a subset of points in the metric, and  each client has a demand, (ii) a vehicle with capacity $Q$, and (iii) a special point $r$
called a depot. A feasible solution of the CVRP is a collection of tours starting from  and ending at the  depot $r$ where on each tour  the vehicle can collect demand from each client such that the total demand collected is at most the capacity $Q$ and over all tours, the demand of every client is collected.   The goal is to find a feasible collection of tours  with minimum cost.  There are three different variants of CVRP: \emph{unit demands} (where the demand of every client is $1$), \emph{splittable demands} (where the demand of each client can be collected in more than one tour), and \emph{unsplittable demands} (where the demand of each client can only be collected in one tour). There has been a tremendous amount of attention recently on this problem (e.g.~\cite{CAFKL20,FKS19}); Yet all the techniques developed only address the special case where  the capacity is $o(\log \log n)$. This stands in stark contrast with the work of  Das and Mathieu~\cite{DasM10}
which gave a QPTAS for the unrestricted problem in Euclidean spaces of fixed dimension in 2010.

\begin{corollary}\label{cor:CVRP}
	For any $1/2 > \eps>0$, there exists a randomized $(1+\epsilon)$-approximation scheme
	for the unit demand Capacitated Vehicle Routing in minor-free metrics with	running time $n^{\eps^{-\Oh(1)}\log^{\Oh(1)} n }$.  
\end{corollary}

The proof follows from Theorem~\ref{thm:main} and the QPTAS of Jayaprakash and Salavatipour~\cite{10.1145/3582500} for bounded treewidth graphs. We refer the readers to \Cref{sec:alg-ap} for more details.

We also obtain the first QPTAS for the classic Facility Location problem in minor-free metrics. Here, the input consists of a metric space, two sets of points $C, F$, and a cost function
$f\colon F \mapsto \mathbb{R}$. The goal is to identify a set $S \subseteq F$ so as to minimize $\sum_{c \in C} \min_{q \in S} \dist(c, q)  + \sum_{q \in S} f(q)$.
While a QPTAS for the problem on Euclidean metrics of fixed dimension has been known since the late 90s~\cite{AroraRR98} and a PTAS later on~\cite{S0097539702404055,CAFS21},
the first PTAS for planar metrics only dates back to 2019~\cite{CPP19} (see also~\cite{CAKM19,Cohen-AddadPP19} for the case with uniform costs).
Yet, an approximation scheme in minor-free metrics has eluded the community until now.

\begin{corollary}\label{cor:FLP}
	For any $1/2 > \eps>0$, there exists a randomized  $(1+\epsilon)$-approximation scheme
	for the Facility Location problem in
	minor-free metrics with
	running time $n^{\eps^{-\Oh(1)}\log^{\Oh(1)} n }$.
\end{corollary}

In the Capacitated $k$-Median problem, we have two sets of points $C, F$ and a capacity constraint $\mu \colon C \mapsto \mathbb{N}$. The goal is to find a set of $k$ centers from $F$ and assignments of the  points of $C$ to the centers such that each center $i \in F$ is assigned at most $\mu(i)$ points, and the total distance from each point to its assigned center is minimum. The  Capacitated Facility Location problem is similar, but now each point  $i$ has an associated opening cost $f(i)$ (for opening a facility at that point) similar to the Facility Location problem above, and the goal is to open a set of facilities $F$  and assignment of other points to each facility so that: (a) each facility in $i \in S$ is assigned at most $\mu(i)$ points, and (b) the total cost $\sum_{i\in S}f(i) + \sum_{p}d(p, S(p))$ is minimum, where $S(p)$ is the facility to which $p$ is assigned, and $d(\cdot, \cdot)$ is the distance function of the input metric. Again, these problems have been
known to admits a QPTAS for Euclidean metrics of fixed dimension~\cite{Cohen-Addad20}, but nothing better than the general metric $\Oh(\log k)$-approximation algorithm was known for
planar (and so minor-free) metrics.

\begin{corollary}\label{cor:clustering}
	For any $1/2 > \eps>0$, there exists a randomized  $(1+\epsilon)$-approximation scheme
	for Capacitated $k$-Median and Capacitated Facility Location in
	minor-free metrics with
	running time $n^{\eps^{-\Oh(1)}\log^{\Oh(1)} n }$.
\end{corollary}

The algorithms in \Cref{cor:clustering} follows the same line of the algorithm for the CVRP: embed the input metric to a low-treewidth graph, and solve the low-treewidth instances in quasi-polynomial time.   

In Euclidean/doubling metrics, many problems were first known to admit QPTASes and then were subsequently improved to PTASes with additional ideas tailoring for each problem.  Following the same pattern, we hope that additional ideas will be developed to improve our QPTASes to PTASes.

\paragraph*{Our techniques.}
For simplicity of exposition, we assume that $\ell = \Oh(\log n)$.
A basic idea in the embedding of Talwar~\cite{Tal04} for metrics of doubling dimension is a randomized ball carving: given a graph $C$ of diameter $D$ in a doubling metric, one can stochastically partition $C$ (using balls of random radii) into a set of clusters, say $\mathcal{P}$, of diameter at most $D/2$ such that the probability of cutting every two points $x$ and $y$ into two different clusters 
is bounded by $\Oh(\frac{\dist(u,v)}{D})$.
Then we construct an $(\eps D/\log n)$-net $N$ of $C$ and designated it as \emph{portals}: in the output host graph, we connect every $x\in N$ to every other point of $C$ by an edge of length equal to the distance of the endpoints in $G$,
and recurse on each cluster in $\mathcal{P}$. Due to the bounded doubling dimension assumption,
the size of $N$ is bounded polynomially in $\log n/\eps$, which finally results in polylogarithmic treewidth. 
For every edge $uv$ in $G$, each of the $\Oh(\log n)$ levels of recursion incurs an expected distortion of
$\Oh(\frac{\eps}{\log n})\dist(u,v)$: the edge $uv$ is cut with probability $\Oh(\frac{\dist(u,v)}{D})$ 
and, thanks to the properties of the chosen net $N$, the detour in the case of cutting $uv$ is bounded by
$\Oh(\frac{\eps}{\log n} D)$. 

One could try to apply Talwar's technique for planar metrics. For the stochastic decomposition into clusters,
the decomposition of  Klein, Plotkin, and Rao~\cite{KPR93} provides the same probability of cutting any two points into two different clusters. However, the step of choosing an $(\eps D/\log n)$-net  $N$ is problematic: we now can no longer bound the size of $N$. In fact, even in as simple graphs as stars it may happen that $|N| = n$. Thus, obtaining any non-trivial treewidth bound following this approach seems very~difficult. On the other hand, stars have already pretty good treewidth.

We thus aims at designing an approach that implicitly distinguishes grid-like planar graphs, to which Talwar's decomposition could be applied, from
planar metrics that are already ``tree-like shaped''.
As obtaining a small $r$-net for $r \ll D$ seems impossible in planar graphs, we aim at lowering the probability that $uv$ is cut
in a way that incurs a detour $\Omega(D)$. 
Consider the following example. Assume that we applied the aforementioned clustering scheme to a graph $C$ and we obtained a partition
$\mathcal{P}$ into clusters that are arranged in a ``grid-like'' fashion: the quotient graph $C/\mathcal{P}$
(defined as a graph with the vertex set $\mathcal{P}$ and two elements $A,B \in \mathcal{P}$ adjacent if $G$ contains an edge between $A$ and $B$)
is a grid. Then, instead of cutting out \emph{all} clusters of $\mathcal{P}$, we would like to do the following: 
choose randomly one of the middle $\xi$ columns of the grid $C/\mathcal{P}$ (for a parameter $\xi$ to be chosen later), carve out the clusters on that column
as subgraphs for recursion, and also recurse on the remaining left and right parts of the grid.
Instead of a net, we designate as portals only one vertex per cluster from the chosen column of the grid. 
In the recursion, the recursive calls for clusters on the chosen column have multiplicatively smaller diameter, while
(with enough columns in the entire grid) the left and right parts have at most $0.75|V(C)|$ vertices;
this gives $\Oh(\ell \log n) = \Oh(\log^2 n)$
bound
on the recursion depth.
For the detour bound, note that now an edge $uv$ has probability $\Oh(\frac{\dist(u,v)}{D})$ of being chosen on 
a boundary of the cluster \emph{times} the probability~$\frac{1}{\xi}$ that this cluster is actually 
in the chosen column. Since an edge that is actually cut experiences a detour of~$\Oh(D)$,
the expected additive distortion due to the detours imposed by the entire process is thus $\Oh(\log^2 n \cdot \frac{1}{\xi} \cdot \frac{\dist(u,v)}{D} D) \sim \Oh(\frac{\log^2 n}{\xi}\dist(u,v))$, and choosing
$\xi$ to be $\log^2 n/\eps$ (which may lead to $\Oh(\xi \log^2 n) = \Oh(\log^4 n/\eps)$ final treewidth bound on the output embedding) gives the desired bound. %

The main argument of our approach is a proof that this set of ``$\xi$ columns to randomly choose from'' can always be found, 
even if $C/\mathcal{P}$ does not look like a grid. 
To this end, we need another property of the clustering algorithm: with high probability,
 the diameter of the quotient graph $C/\mathcal{P}$ is bounded by a polylogarithmic function of $n$, which we henceforth denote
 by $\rho$. 
Then, the treewidth of $C/\mathcal{P}$ is bounded by $\Oh(\rho)$, which gives us the first ``column'': a set $\Aa_1$ of $\Oh(\rho)$
of clusters, whose deletion leaves connected components of multiplicatively smaller size. We henceforth refer to
such a set of clusters as a \emph{balanced cut}.

We need more balanced cuts to choose from and we would like the balanced cuts to be mostly disjoint in the sense of
not containing the same edge of $G$ on their boundary many times.
However, there is a number of star-like examples where this is not possible, i.e., there is a cluster contained
in every balanced cut. We proceed as follows: after picking the first balanced cut $\Aa_1$ as in the previous paragraph,
we refine te partition $\mathcal{P}$ by, for every cluster $C_1 \in \Aa_1$, 
replacing $C_1$ with the result of applying again the clustering scheme to $C_1$ and diameter $D/2$. 
Then, we find a balanced cut $\Aa_2$ in the obtained partition $\mathcal{P}_2$, and continue to $\mathcal{P}_3$
by again replacing the clusters used by $\Aa_2$ with a result of applying the clustering scheme to them. 
The replacement process stop at single-vertex clusters, where cutting would not cause any detour. 

Now, every edge $uv$ is contained in the boundary of clusters larger than single vertices in $\Oh(\ell) = \Oh(\log n)$
balanced cuts $\Aa_i$. Hence, if we produce $\Omega(\frac{\log^3 n}{\eps})$ balanced cuts, the expected detour
analysis will still give the desired result. But there is a problem: as we expand the partitions $\mathcal{P}_i$,
the diameter and hence the bound on the treewidth of $C/\mathcal{P}_i$ grows, so our guarantee on the size of $\Aa_i$ detoriates
as $i$ increases.

This is the moment where the topological assumption on the input graph finally comes into play.
Naively, the difference between the diameters of $C/\mathcal{P}_{i+1}$ and $C/\mathcal{P}_i$ can be as large
as $\Omega(|\Aa_i| \cdot \rho)$, as every cluster of $\Aa_i$ is replaced with a set of clusters whose quotient
graph has diameter $\rho$. However, if we assume that the given graph excludes some fixed apex graph\footnote{An {\em{apex graph}} is a graph that can be made planar by removing one vertex.} as a minor, then we can prove that after $k$ steps, the diameter of $C/\mathcal{P}_k$ is only
$\Oh(\sqrt{\sum_{i=1}^{k-1} |\Aa_i|}\cdot \ell \rho)$. The key graph theoretic property of apex-minor-free graphs used here is the following: to cover an apex-minor-free (unweighted) graph of treewidth $\geq t$ with balls of radius $d$, one has to pick at least $\Omega((t/d)^2)$ such balls.
This proves that if one continues the process of finding and unraveling cuts for $\xi$ steps, where $\xi$ is any large polylogarithm one wishes for, 
all created balanced cuts will have sizes bounded by another polylogarithmic bound $\tau$ depending only on $\xi$. 
This, in turn, gives the final polylogarithmic bound on the treewidth of the obtained embedding.

This concludes the overview of the proof of Theorem~\ref{thm:main} under the assumption that the input graph excludes a fixed apex graph as a minor; this in particular includes planar graphs, which exclude $K_5$.
To extend this to the setting when any fixed graph $K$ is excluded, we start by applying the Robertson-Seymour structure
theorem~\cite{gm16} that essentially states that the input graph $G$ admits a tree decomposition where the torso
of every bag consists of a graph coming from a fixed apex-minor-free graph class, with a constant number of vertices added to it.
We then carefully apply our technique to the apex-minor-free part of the centroid bag of that decomposition.

\paragraph*{A bit more on related works.}
 There is a long line of work on embedding planar and minor-free metrics into Euclidean and normed spaces.  Rao~\cite{Rao99} constructed an embedding of planar metrics into Euclidean spaces with distortion $\Oh(\sqrt{\log n})$.  For any $\ell_p$ space with $p\geq 1$, Rao obtained an embedding with $\Oh((\log n)^{1/p})$ distortion.  Rao's results were extended to  minor-free metrics~\cite{AGGNT19,KLMN04}; the distortion is $\Oh(h^{1-1/p}(\log n)^{1/p})$ where $h$ is the size of the excluded minor. Better embedding results were known for graphs of small pathwidth~\cite{AFGN22,LS13}. For the Euclidean case,  Newman and Rabinovich~\cite{NR02} provided a matching lower bound for Rao's embedding, showing that distortion $\Oh(\sqrt{\log n})$ achieved by Rao is the best possible.  
 
 Of special interest is the case of $\ell_1$. The famous $\ell_1$ embedding conjecture~\cite{GNRS04} states that planar metrics  are embeddable into $\ell_1$ with $\Oh(1)$  distortion. The conjecture is  believed to hold for general minor-free metrics~\cite{GNRS04,LS09}. While the conjecture remains open even for planar metrics, progress has been made   for various special cases: $k$-outerplanar graphs~\cite{GNRS04}, graphs with a small face cover~\cite{KLR19,Filtser20}, or when every demand pair has a face containing both endpoints~\cite{Kumar2022}. We refer readers to~\cite{AFGN22} and  references therein for more embedding results.

 \paragraph*{Organization.}
We present our clustering tool in Section~\ref{sec:clustering}.
Sections~\ref{sec:tw} and~\ref{sec:embed} prove Theorem~\ref{thm:main} for the case of graph classes excluding a fixed apex graph as a minor.
Then, in Section~\ref{sec:minor}, we discuss how to extend the proof to graphs excluding an arbitrary fixed minor.
Section~\ref{sec:alg-ap} provides details on the aforementioned applications of Theorem~\ref{thm:main}.

\newcommand{\bag}{\mathsf{bag}}

\section{Preliminaries}

For a real $t$, we sometimes write $\exp(t)\coloneqq e^t$, where $e$ is the base of the natural logarithm. We use the Iverson bracket notation: for a boolean assertion $\psi$, $[\psi]$ is $1$ if $\psi$ holds and $0$ otherwise.

\paragraph*{Graphs.} We use standard graph notation. Most of graphs considered in this paper are edge-weighted, which means that with every edge $e$ we associate a positive weight $\len(e)$, interpreted as the length of $e$. Hence, for an edge-weighted graph $G$ we can consider the {\em{distance metric}} induced in $G$: for vertices $u,v$, $\dist_G(u,v)$ is the smallest possible length of a path connecting $u$ and $v$, or $+\infty$ if such a path does not exist, where the {\em{length}} of a path $P$, denoted $\len(P)$, is the sum of the weights of edges of $P$. The {\em{stretch}} of this metric is the ratio between the maximum finite distance and the minimum distance between two distinct vertices in $G$.

The notions of diameter and radius in an edge-weighted graph $G$ are defined with respect to the distance metric induced in $G$. When speaking about distances in an unweighted graph, or about its diameter or radius, we mean the standard graph-theoretic notions, where every edge is considered to have length $1$. 

For a graph $G$ and a partition $\partition$ of the vertex set of $G$, we define the {\em{quotient graph}} $G/\partition$ as the unweighted graph on vertex set $\partition$ where parts $C,C'\in \partition$ are adjacent if in $G$ there is an edge with one endpoint in $C$ and the other in $C'$. Note that while $G$ can be edge-weighted, we always consider $G/\partition$ to be an unweighted graph.

\paragraph*{Decompositions.}
A {\em{tree decomposition}} of a graph $G$ is a tree $T$ together with a function $\bag(\cdot)$ that with every node $x$ of $T$ associates its bag $\bag(x)\subseteq V(G)$ so that (i) for every vertex $u$ of $G$, the nodes of $T$ whose bags contain $u$ form a connected subtree of $T$, and (ii) whenever $u$ and $v$ are adjacent in $G$, there is a node in $T$ whose bag contains both $u$ and $v$. The {\em{width}} of a tree decomposition is the maximum size of a bag minus $1$, and the {\em{treewidth}} of a graph $G$ is the minimum possible width of a tree decomposition of $G$.

An {\em{elimination forest}} of a graph $G$ is a rooted forest $F$ with the same vertex set as $G$ satisfying the following property: whenever $u$ and $v$ are adjacent in $G$, either $u$ is an ancestor of $v$ in $F$ or vice versa. The {\em{treedepth}} of a graph $G$ is the minimum possible depth (number of vertices on a root-to-leaf path) of an elimination forest of $G$. 

It is easy to see that the treewidth of a graph is never larger than the treedepth minus $1$: to construct a tree decomposition $(T,\bag)$ of width $d-1$ from an elimination forest $F$ of depth $d$, take $T$ to be~$F$ with all roots connected into a path, and for every vertex $u$ set $\bag(u)$ to be the set of ancestors of $u$ in~$F$. On the other hand, it is well-known that the treedepth of an $n$-vertex graph is never larger than $\log_2 n$ times its treewidth plus $1$~\cite[Corollary~6.1]{sparsity}. Thus, from the point of view of constructing metric embeddings into graphs of polylogarithmic treewidth, we can equivalently consider embeddings into graphs of polylogarithmic treedepth instead.

We use the standard minor order for graphs. A graph $K$ is an {\em{apex graph}} if there is a vertex $u$ such that $K-u$ is planar. The following lemma, which follows from the previous literature, is the sole argument in which we exploit the topological simplicity of the considered graphs.

\begin{lemma}\label{lem:tw:balls}
For every apex graph $K$ there is a positive integer $\alpha_K$ such that the following holds. Suppose $G$ is a $K$-minor-free (unweighted) graph and $Z$ is a subset of vertices of $G$ such that every vertex of $G$ is at distance at most $d$ from a vertex of $Z$. Then the treewidth of $G$ is upper bounded by $\alpha_K d \sqrt{|Z|}$.
\end{lemma}

Lemma~\ref{lem:tw:balls} was proved explicitly in the case of planar graphs by Demaine et al.~\cite[Theorem~3.3]{DemaineFHT05}. In their proof, they use only two properties of planar graphs: (a) if a planar graph $G$ has treewidth $t$, then $G$ contains the $s\times s$ grid as a minor for $s\in \Omega(t)$, and (b) if a planar graph $G$ contains the $s\times s$ grid as a minor, then every {\em{distance-$d$ dominating set}} in $G$ --- a set $Z$ such that every vertex of $G$ is at distance at most $d$ from $Z$ --- has size $\Omega((s/d)^2)$ (see~\cite[Lemma~3.1]{DemaineFHT05}). To lift this proof to the case of $K$-minor-free graphs for a fixed apex graph $K$, one can use the result of Fomin et al.~\cite[Theorem~3]{FominGT11}: If a $K$-minor-free graph $G$ has treewidth $t$, then $G$ contains the {\em{triangulated $s\times s$ grid}} $\Gamma_s$ as a contraction, for $s\in \Omega_K(t)$ (see~\cite{FominGT11} for a definition of $\Gamma_s$). It can be easily seen that the minimum size of a distance-$d$ dominating set in a graph does not increase under edge contractions, and is of the order $\Omega((s/d)^2)$ in~$\Gamma_s$. Therefore, every distance-$d$ dominating set in a $K$-minor-free graph of treewidth $t$ must have size $\Omega_K((t/d)^2)$; this readily implies the statement of Lemma~\ref{lem:tw:balls}.

\paragraph*{Metric embeddings.} A {\em{metric embedding}} of an edge-weighted graph $G$ consists of the {\em{host graph}} --- an edge-weighted graph $H$ --- and the {\em{embedding}} $\eta\colon V(G)\to V(H)$ satisfying $\dist_H(\eta(u),\eta(v))\geq \dist_G(u,v)$ for all $u,v\in V(G)$. The embedding $\eta$ has {\em{(multiplicative) distortion}} $c$ if $\dist_H(\eta(u),\eta(v))\leq c \cdot \dist_G(u,v)$ for all $u,v\in V(G)$. It will be often the case that a metric embedding is constructed by a randomized procedure; then we can say that $\eta$ has {\em{expected (multiplicative) distortion}}~$c$ if $\Exp\ \dist_H(\eta(u),\eta(v))\leq c \cdot \dist_G(u,v)$ for all $u,v\in V(G)$.

\section{Partitioning tools}\label{sec:clustering}

In this section we introduce our probabilistic partitioning tools, which are heavily inspired by the work on {\em{padded decompositions}}~\cite{AbrahamBN06,FRT04}. We start by presenting a procedure that partitions a given graph into clusters of significantly smaller diameter so that for any path $Q$, the expected number of edges of $Q$ cut by the boundaries of the clusters is proportional to the length of $Q$. We then use our partitioning procedure in a hierarchical fashion to create a multi-level decomposition of a given~graph.

\subsection{Probabilistic preliminaries}

We need a few standard tools from the probability theory. 
We will use the standard Chernoff bound, stated below, and its simple
corollary, proven for completeness in Appendix~\ref{app:boring}.

\begin{theorem}[Multiplicative Chernoff bound]\label{thm:MultChernoff}
Let $X_1,\ldots,X_n$ be independent random variables with values in $\{0,1\}$.
Let $X = \sum_{i=1}^n X_i$, let $\mu = \Exp[X]$, and let $0\le  \delta$ be a real.
Then, 
  \[ \Prob \left(X \geq \mu (1 + \delta)\right) \leq \exp(-\delta^2 \mu/(2+\delta) ). \]
\end{theorem}

\begin{lemma}\label{lem:hit-process}
Let $k \geq 1$ be an integer and let $X_1,X_2,\ldots$ be independent random variables with values $\{-1,1\}$ such that
for every $i \geq 1$, $\Prob(X_i = 1) \leq \frac{1}{8}$. 
Let $Z$ be a random viariable defined as the minimum $i$ such that $k + \sum_{j=1}^i X_j < 0$. 
Then 
\[ \Prob\left(Z > 2k\right) \leq \exp\left(-\frac{k}{12}\right).\]
\end{lemma}

We will also use random variables with exponential distributions. Recall that a real-valued random variable $X$ has exponential distribution with parameter $\lambda>0$, denoted $X\sim \ExpDist(\lambda)$, if its distribution has density given $g(t)=\lambda \exp(-\lambda t)\cdot [t\geq 0]$. It is well-known that exponentially-distributed random variables are {\em{memoryless}} in the sense explained by the statement below.

\begin{lemma}\label{lem:memoryless}
 Suppose $X$ is a random variable with an exponential distribution. Then for every pair of nonnegative reals $s\leq t$, we have
 \[\Prob(X\leq s+t~|~X\geq s) = \Prob(X\leq t).\]
\end{lemma}

In one of the more involved proofs, we will use memorylessness in a quite convoluted fashion.
In order not to clutter the arguments later, we abstract the probabilistic part as the lemma below. 
While this statement may be next to trivial for readers more experienced with the probability theory in general 
or exponential clocks in particular, for the sake of rigour we provide a proof in Appendix~\ref{app:boring}.

\begin{lemma}\label{lem:memory-application}
Let $n$ be a positive integer, $X_1,X_2,\ldots,X_{2n}$ be independent random variables, each with distributon $\ExpDist(1)$, let $f_1,f_2,\ldots,f_{2n}$ be measurable functions where $f_i \colon \R^{i-1} \to \R_{\geq 0} \cup \{+\infty\}$ for $1 \leq i \leq 2n$
and $f_i \equiv 0$ for $n < i \leq 2n$, and let $p \geq 0$ be a real.
We say that an index $i$ is \emph{good} if $X_i \geq f_i(X_1,X_2,\ldots,X_{i-1})$ and \emph{excellent} if additionally $X_i > f_i(X_1,X_2,\ldots,X_{i-1}) + p$.
Let $N \in \{n,n+1,\ldots,2n\}$ be the number of good indices (note that every $i>n$ is good almost surely) and let $(I_i)_{i=1}^{N}$ be random variables being consecutive good indices.
For $1 \leq i \leq N$, let $Y_i = X_{I_i} - f_{I_i}(X_1, \ldots, X_{I_i-1})$. 
Then, $Y_1, Y_2, \ldots, Y_n$ are independent random variables, each with distribution $\ExpDist(1)$. 
Consequently, if for $1 \leq i \leq n$ we denote by $A_i$ the event ``$I_i$ is excellent'', then the events $(A_i)_{i=1}^n$ are independent and each happens with probability~$e^{-p}$.
\end{lemma}

\subsection{A single level}

We first define a simple partitioning procedure to create one level of a hierarchical decomposition. The properties of the procedure are asserted in the following lemma.

\begin{lemma}\label{lem:single-level}
 There is a randomized polynomial-time algorithm that given a connected $n$-vertex edge-weighted graph $G$ with diameter $D$ and a real parameter $r>0$, outputs a partition $\partition$ of the vertex set of $G$ into {\em{clusters}} satisfying the following properties:
 \begin{enumerate}[label=(P\arabic*),ref=(P\arabic*),leftmargin=*,nosep]
  \item\label{p:conn} For each $C\in \partition$, $G[C]$ is connected.
  \item\label{p:radii} For every $t\geq 0$, the probability that in $\partition$ there is a cluster inducing a subgraph of diameter larger than $2r(t + 1 + \ln n)$ is at most $e^{-t}$.
  \item\label{p:quo} For every $t\geq 0$, the probability that the diameter of $G/\partition$ is
larger than $120\cdot \frac{D}{r}\cdot (t+1+\ln n)$ is at most $e^{-t}$.
  \item\label{p:geo} For every path $Q$ in $G$, the expected number of edges on $Q$ whose endpoints belong to different clusters is at most $\frac{\len(Q)}{D}\cdot 4\ell$, where $\ell$ is the least integer such that the stretch of the distance metric in $G$ is strictly smaller than $2^\ell$.
\end{enumerate}
\end{lemma}

The remainder of this section is devoted to proving Lemma~\ref{lem:single-level}. For tie-breaking, we assume that there is a fixed in advance total order $\preceq$ on the vertices of $G$.  By scaling, we may assume that $D=2^\ell$ and that the distance between every pair of distinct vertices of $G$ is strictly larger than $1$. Further, we may assume that $r<D$, for otherwise we may return a partition consisting of one part: the whole~$V(G)$.

\paragraph{The procedure.} At the beginnning, the procedure initiates $\partition$ to be $\emptyset$, and it will iteratively add clusters (subsets of $V(G)$) to $\partition$ one-by-one until $\partition$
becomes a partition of $V(G)$. At every step, a vertex is called \emph{free} if it does not belong to any element of $\partition$.

The procedure iteratively, as long as there exists a free vertex, picks a $\preceq$-minimal  free vertex $v$,
  samples a radius $r_v \coloneqq r(1+X_v)$, where $X_v\sim \ExpDist(1)$ is an independent random variable with exponential 
  distribution, and creates a new element $C_v$ in $\partition$
  consisting of all vertices within distance at most $r_v$ from $v$ in the graph $G-\bigcup \partition$ 
  (i.e., subgraph of $G$ induced by the free vertices). 
  Vertex $v$ is called the \emph{center} of $C_v$, set $C_v$ is the \emph{cluster} of $v$,
  and $r_v$ is the \emph{radius} of $v$.

\paragraph{Running time.}
The procedure clearly takes (randomized) polynomial time and produces a partition $\partition$ of $V(G)$ into 
clusters. Note that every cluster $C_v$ is connected, contains $v$, and is of radius at most $r_v$.

\paragraph{Properties.} We now verify properties \ref{p:conn}--\ref{p:geo}.
Let $k$ bet the number of steps taken by the procedure, and let $v_1,\ldots,v_k$ the centers of the consecutive clusters constructed by the procedure.
For brevity, we denote $X_i\coloneqq X_{v_i}$, $C_i \coloneqq C_{v_i}$, and $r_i \coloneqq r_{v_i}$, for $i\in \{1,\ldots,k\}$. Furthermore, let
$G_i = G - \bigcup_{j < i} C_j$; that is, $G_i$ is the subgraph of $G$ induced by the vertices free at the moment of creation of $C_i$.

Property~\ref{p:conn} follows directly from the construction. Since the diameter of $G[C_i]$ is at most $2r_i$, property~\ref{p:radii} follows directly from the following statement.

\begin{claim}\label{lem:part:radii-bound}
For every $t \geq 0$, the probability that there exists $i\in \{1,\ldots,k\}$
such that $r_i > r(t + 1 + \ln n)$ is at most $e^{-t}$. 
\end{claim}
\begin{proof}
For every $i\in \{1,\ldots,k\}$, we have
\[ \Prob\left(r_i > r(t + 1+\ln n)\right) = \Prob(X_i > t + \ln n) = e^{-(t+\ln n)} = e^{-t}/n.\]
The lemma follows by the union bound, as there are at most $n$ clusters.
\end{proof}

In the next claim we verify the key property~\ref{p:quo}. For brevity, let $H\coloneqq G/\partition$ be the quotient graph. Recall that $H$ is unweighted.

\begin{claim}\label{lem:part:diam}
  For every $t\geq 0$, the probability that the diameter of $H$ is
larger than $120\cdot\frac{D}{r}\cdot (t+1+\ln n)$ is at most $e^{-t}$.
\end{claim}
\begin{proof}
  Consider any two vertices $u_1$ and $u_2$. Let $Q$ be any shortest path from $u_1$ to $u_2$ in $G$ (with respect to the edge weights).
  By definition we have $|Q| \leq D$.
  In the argumentation below, we consider $Q$ to be oriented from $u_1$ to $u_2$.

Let $Z = \{v \in V(H)~|~C_v \cap V(Q) \neq \emptyset\}$ be the set of those centers whose clusters intersect $Q$. We color the vertices of $Q$ using palette $Z$ as follows.
Initially, all vertices are uncolored. We consider the vertices $v \in Z$ in the same order as the clustering procedure.
Upon considering vertex $v$, we find the first and the last vertex of $C_v\cap V(Q)$ that are uncolored so far, say $q_v^1$ and $q_v^2$, and we color with color $v$ \emph{all} vertices lying between $q_v^1$ and $q_v^2$ on $Q$. Note that this may result in recoloring some vertices that were colored in previous iterations.
If all vertices of $C_v\cap V(Q)$ were colored before, we do not color any vertices in this iteration, hence there will be no vertex of color $v$.

Note that 
at every moment, all vertices colored with any color $v\in Z$ form an interval on $Q$.
Furthermore, if when considering a center $v$ we recolor at least one vertex colored $v'$ to $v$, then the interval between $q_v^1$ and $q_v^2$ on $Q$ contains {\em{all}} vertices colored previously with $v'$, hence \emph{all} vertices colored $v'$ get recolored to $v$.
Finally, note that at the end of the procedure, every vertex is assigned some color.

Let $Z'' \subseteq Z$ be the set consisting of all vertices $v \in Z$ such that when $v$ was considered during the coloring procedure, at least one vertex of $Q$ was colored with color $v$. Further,
 let $Z' \subseteq Z''$ be the set consisting of all those vertices $v \in Z$ for which at least one vertex of $Q$ is colored with color $v$ at the end of the coloring procedure.
 Let $\wh{u}_1$ and $\wh{u}_2$ be cluster centers such that $u_1\in C_{\wh{u}_1}$ and $u_2\in C_{\wh{u}_2}$. 
Observe that in $H$ there is a path from $\wh{u}_1$ to $\wh{u}_2$ that uses only vertices of $Z'$, as for every $v \in Z'$, both $q_v^1$ and $q_v^2$ are contracted onto $v$ in $H$.
Hence, 
$$\dist_H(\wh{u}_1,\wh{u}_2)\leq |Z'|-1.$$
So to give an upper bound on the probability that $\dist_H(\wh{u}_1,\wh{u}_2)$ is large, it suffices to bound the probability that $Z'$ is large.

Consider the moment at the beginning of the $i$-th step of the procedure, that is, when clusters $C_j$ for $j < i$ were already
constructed. A \emph{segment} is a maximal subpath of $Q$ whose all vertices are uncolored; let $\mathcal{S}_i$
be the family of segments at the beginning of the $i$th step.
Call the cluster $C_i$ \emph{tenoning} if $C_i$ contains all vertices of at least one segment. 

Let
\[\lambda\coloneqq t+1+\ln n.\]
For any center $v_i$, let $R_i$ be the distance in $G_i$ from $v_i$ to the nearest vertex that is uncolored in~$G_i$.
We have $v_i \in Z''$ if and only if $r_i \geq R_i$. 
Let $p\coloneqq \frac{1}{12\lambda}$; note that thus, $p\leq \frac{1}{12}\leq \frac{1}{8}$.
We say that the vertex $v_i$ is \emph{good} if $r_i > R_i + pr$.
The crucial observation is the following: if $v_i$ is good, then it is either tenoning, 
or there is a subpath $Q'$ of $Q$ of length at least $pr$ such that all vertices of $Q'$ were uncolored before the $i$th step, and all of them, except possibly for the endpoints, receive color $v_i$ in the $i$th step. In the former case we have
\[|\mathcal{S}_{i+1}| < |\mathcal{S}_i|,\]
and in the latter case we have 
\[ \sum_{S \in \mathcal{S}_{i}} \len(S) - \sum_{S \in \mathcal{S}_{i+1}} \len(S) \geq pr. \]

This motivates the following measurement.
The \emph{potential} at the beginning of step $i$ is defined as:
\[ \pi_i \coloneqq |\mathcal{S}_i| + 2\cdot  \frac{ \sum_{S \in \mathcal{S}_i} \len(S)}{pr}. \]
Note that 
\[ \pi_1 \leq 1 + \left\lceil\frac{2D}{pr}\right\rceil \eqqcolon k. \]
Clearly, $|\mathcal{S}_{i+1}| \leq |\mathcal{S}_i| + 1$ and $\bigcup \mathcal{S}_{i+1} \subseteq \bigcup \mathcal{S}_i$, for all relevant $i$.
Hence, $\pi_{i+1} \leq \pi_i + 1$. Furthermore, if $v_i$ is good, then $\pi_{i+1} \leq \pi_i - 1$.
Note that $\pi_i > 0$ for $1 \leq i \leq |Z''|$.

Let $I_1,I_2,\ldots,I_{|Z''|}$ be the indices of the consecutive elements of $Z''$ (note that these are random variables,
    and even the length of this sequence is a random variable). 
For $j \geq 1$, let $Y_j \coloneqq (r_{I_j} - R_{I_j})/r$.
For $j \geq 1$, let $A_j$ be the event ``$v_{I_j}$ is good'', which is the same as ``$Y_j > p$''.

Now, an important observation is that Lemma~\ref{lem:memory-application} applies:
variables $Y_j$ are independent, each has distribution $\ExpDist(1)$, and events $(A_j)_{j \geq 1}$ are independent
and each individually happens with probability $e^{-p}$.
Indeed, to formally apply Lemma~\ref{lem:memory-application} we artifically extend the sequence $X_1,X_2,X_3,\ldots = X_{v_1},X_{v_2},X_{v_3},\ldots$ of independent
random variables with distribution $\ExpDist(1)$ to length $2n$
and set $f_i \coloneqq R_i / r$ whenever $i$ is still an index of a center and $f_i \equiv 0$ otherwise.

Then, as $1 - e^{-p} \leq p \leq \frac{1}{8}$ and $r < D$, by Lemma~\ref{lem:hit-process} we infer that
\[ \Prob(|Z''|>2k) \leq \exp\left(-\frac{k}{12} \right)
 \leq \exp\left(-\frac{D}{6pr}\right) \leq \exp\left(-2\lambda\right)= \exp\left(-2(t+1+\ln n)\right).\]
On the other hand, using $r < D$,
  \[ 2k = 2+2\left\lceil\frac{2D}{pr}\right\rceil \leq \frac{D}{r} \cdot 10 \cdot 12\lambda = 120\cdot \frac{D}{r}\cdot (t+1+\ln n).\]
Denote $w\coloneqq 120\cdot \frac{D}{r}\cdot (t+1+\ln n)$. This means that
\begin{align*}
 \Prob(\dist_H(\wh{u}_1,\wh{u}_2)>w)& \leq \Prob(|Z'|>w)\leq \Prob(|Z''|>w)\\
 &\leq \Prob(|Z''|>2k)  \leq \exp\left(-2(t+1+\ln n) \right).
\end{align*}
By the union bound over all choices of $u_1$ and $u_2$, we obtain
\[ \Prob\left( \mathrm{diam}(H) > w \right) \leq 
\exp\left( - 2(t+1+\ln n)+2\ln n \right)\leq e^{-t}. \]
This finishes the proof.
\end{proof}

Finally, we show property~\ref{p:geo}.

\begin{claim}\label{lem:part:cut}
Let $Q$ be a path in $G$.
Then the expected number of edges of $Q$ whose endpoints belong to different clusters in $\partition$
is bounded by $\frac{\len(Q)}{r} \cdot 4\ell$.
\end{claim}
\begin{proof}
We color the  edges of $Q$ with colors being cluster centers as follows:
for an edge $xy$ on $Q$, if $x \in C_u$ and $y \in C_v$, then the edge $xy$ is colored with the center among $\{u,v\}$ whose cluster was conceived earlier (so if $u=v$, the edge $xy$ is colored $v$). 
A \emph{segment} is a maximal subpath of edges of $Q$ of the same color.
An edge of $Q$ with endpoints belonging to different clusters is \emph{charged} to the segment it is contained in. Note that 
a segment can be charged by $0$, $1$, or $2$ edges; a segment is \emph{positive} if it is charged by at least one edge.

We now make a crucial observation: For every edge $e$ on $Q$ and every real $d > 0$,
\begin{equation}\label{eq:part:cut}
\Prob(e\textrm{ lies in a positive segment of length }\leq d) \leq \frac{d}{r}.
\end{equation}
Indeed, for an index $i$, let $x_i$ be the endpoint of $e$ that is closer to $v_i$ in $G_i$ (or $x_i = \bot$ if both endpoints of $e$ are already in clusters)
and let $R_i$ be the distance between $v_i$ and $x_i$ in $G_i$ ($R_i = +\infty$ if $x_i = \bot$).
Note that $R_i$ is a measurable function of values $(X_j)_{j<i}$.
Let $I$ be the index such that the color of $e$ is $v_I$; equivalently, $I$ is the earliest index such that $r_I\geq R_I$.  (Note that $I$ is a random variable.) 
Next, let $S$ be the segment of $e$, and suppose for a moment that $S$ is positive.
Let $yz$ be an edge of $S$ that is charged to $S$, with $y \in C_I$ and $z \notin C_I$. Note that both $y$ and $z$ are free at the moment when $C_I$ is constructed. Since $z\notin C_I$, we have $\dist_{G_I}(v_I,z)>r_I$, which together with $\dist_{G_I}(v_I,x_I)= R_I$ implies that
$$\len(S)\geq \dist_{G_I}(x_I,z)\geq r_I-R_I.$$
We conclude that the event ``$e$ lies in a positive segment of length $\leq d$'' is contained in the event $R_I\leq r_I\leq R_I+d$. We may now use Lemma~\ref{lem:memory-application} to conclude that this event has probability at most $1-e^{-d/r}\leq d/r$; this establishes~\eqref{eq:part:cut}.

For an edge $e$ of $Q$, let $S(e)$ be the segment containing $e$. 
We bound the number of positive segments as follows.
\begin{align*}
\#(\textrm{positive segments}) &\leq \sum_{e \in E(Q)} \frac{\len(e) \cdot [S(e)\textrm{ is positive}]}{\len(S(e))} \\
    &\leq \sum_{e \in E(Q)} \sum_{i=1}^{\ell} \len(e) \cdot [S(e) \textrm{ is positive}] \cdot [2^{i-1} \leq \len(S(e)) \leq 2^i] \cdot 2^{-i+1} \\
    &\leq 2 \sum_{i=1}^{\ell} \sum_{e \in E(Q)} \len(e) \cdot [S(e) \textrm{ is positive}] \cdot [\len(S(e)) \leq 2^i] \cdot 2^{-i}.
\end{align*}
Hence, by~\eqref{eq:part:cut}, the expected number of positive segments is bounded by
\[ 2 \sum_{i=1}^{\ell} \sum_{e \in E(Q)} \len(e) \cdot \frac{2^i}{r} \cdot 2^{-i} = \frac{2 \ell}{r} \sum_{e \in E(Q)} \len(e) = \frac{\len(Q)}{r} \cdot 2\ell. \]
This finishes the proof, as the total number of edges of $Q$ whose endpoints belong to different clusters is at most twice the number of positive segments. 
\end{proof}

As all the required properties are argued, the proof of Lemma~\ref{lem:single-level} is complete.

\subsection{Hierarchical decomposition}\label{sec:decomp}

We will use the decomposition procedure described in the previous section recursively,
to construct a hierarchical decomposition of a given graph into smaller and smaller induced subgraphs of the input graph. To formulate the properties of the hierarchical decompositions, we need a few definitions. 

A {\em{clustering chain}} of a connected graph $G$ (not necessarily edge-weighted) is a sequence
$\partition_0,\partition_1,\ldots,\partition_k$ of partitions of $V(G)$ such that
\begin{itemize}[nosep]
 \item $\partition_k=\{V(G)\}$;
 \item $\partition_0$ is the discrete partition with every vertex in a separate part; and
 \item for all $i,j\in \{0,1,\ldots,k\}$ with $i<j$, $\partition_i$ {\em{refines}} $\partition_j$ in the following sense: every part of $\partition_i$ is entirely contained in a single part of $\partition_j$;
 \item for all $i\in \{0,1,\ldots,k\}$ and each $C\in \partition_i$, the graph $G[C]$ is connected.
\end{itemize}
For $i\in \{0,1,\ldots,k-1\}$ and a part $C\in \partition_{i+1}$, we let $\partition_i[C]$ be the set of parts of $\partition_i$ that are entirely contained in $C$; note that $\partition_i[C]$ is a partition of $C$.
For an edge $e$ of $G$, the {\em{level}} of $e$ is the largest index $i\in \{1,\ldots,k\}$ such that the endpoints of $e$ belong to different parts of $\partition_i$.

With these definition in place, we can formulate a statement that summarizes the properties of the hierarchical decomposition.

\begin{lemma}\label{lem:hierarchical}\label{lem:hierarchical:stmt}
 There is a randomized polynomial-time algorithm that given a parameter $\delta>0$ and a connected edge-weighted $n$-vertex graph $G$ where the distance between any pair of vertices is in the interval $(1,2^\ell]$, either reports failure or computes a clustering chain $\partition_0,\partition_1,\ldots,\partition_\ell$ of $G$ with the following properties:
 \begin{enumerate}[label=(Q\arabic*),ref=(Q\arabic*),leftmargin=*]
  \item\label{p:hiercl}\label{p:stmt:hiercl} For every $i\in \{0,1,\ldots,\ell\}$ and every cluster $C\in \partition_i$, the graph $G[C]$ has diameter at most $2^i$.
  \item\label{p:hierdiam}\label{p:stmt:hierdiam} For every $i\in \{0,\ldots,\ell-1\}$ and every cluster $C\in \partition_{i+1}$, the quotient graph $G[C]/\partition_i[C]$ has diameter at most $480(\ln(2\ell n^2/\delta)+1)^2$.
  \item\label{p:level}\label{p:stmt:level} For every path $Q$ in $G$ and $i\in \{0,\ldots,\ell-1\}$, the expected number of edges of $Q$ of level $i$ is upper bounded by $2^{-i}\cdot 8\ell(\ln(2\ell n^2/\delta)+1)\cdot \len(Q)$.
 \end{enumerate}
 The algorithm reports failure with probability at most $\delta$.
\end{lemma}

The remainder of this section is devoted to the proof of Lemma~\ref{lem:hierarchical}.

The algorithm constructing the chain $\partition_0,\partition_1,\ldots,\partition_\ell$ is very simple. First, set $\partition_\ell\coloneqq V(G)$. Then, iteratively for every $i=\ell-1,\ldots,1$, consider the graph $G[C]$ for every part $C\in \partition_{i+1}$, and apply the algorithm of Lemma~\ref{lem:single-level} to $G[C]$ with parameter 
\[r_i\coloneqq \frac{2^{i-1}}{\ln(2\ell n^2/\delta)+1}.\] 
Thus we obtain a partition $\partition_i[C]$ of $C$, and we set $\partition_i\coloneqq \bigcup_{C\in \partition_{i+1}} \partition_i[C]$.

Call the resulting sequence of partitions $\Ll\coloneqq (\partition_0,\partition_1,\ldots,\partition_\ell)$ {\em{good}} if the following conditions hold:
\begin{itemize}[nosep]
 \item For each $i\in \{0,\ldots,\ell\}$ and $C\in \partition_i$, the graph $G[C]$ has diameter at most $2^i$.
 \item For each $i\in \{0,\ldots,\ell-1\}$ and each cluster $C\in \partition_{i+1}$, the quotient graph $G[C]/\partition_i[C]$ has diameter at most $480(\ln(2\ell n^2/\delta)+1)^2$.
\end{itemize}
Note that if $\Ll$ is good, then it is a clustering chain. Indeed, then for every cluster $C\in \partition_0$, $G[C]$ has diameter at most $1$, so since distances in $G$ are always strictly larger than $1$, $C$ cannot have more than one vertex. All other properties of a clustering chain follow directly from the construction.

We now verify that $\Ll$ is good with high probability.

\begin{claim}
We have
\[\Prob(\Ll\textrm{ is not good})\leq \delta.\]
\end{claim}
\begin{proof}
 Observe that $\Ll$ is good provided the in each application of the algorithm of Lemma~\ref{lem:single-level}, the events described is properties~\ref{p:radii} and~\ref{p:quo} do not happen, where in both properties we set \[t\coloneqq \ln(2\ell n/\delta).\]
 Indeed, when we apply the algorithm to $G[C]$ for some cluster $C\in \partition_{i+1}$ ($i\in \{0,1,\ldots,\ell-1\}$), the two events described in properties~\ref{p:radii} and~\ref{p:quo} can be rephrased as follows:
 \begin{itemize}[nosep]
  \item Property~\ref{p:radii}: There is a cluster $C'\in \partition_i[C]$ such that $G[C']$ has diameter larger than
  \[2r_i(t+1+\ln n)=2^i.\]
  \item Property~\ref{p:quo}: The quotient graph $G[C]/\partition_i[C]$ has diameter larger than
  \[120\cdot \frac{D}{r_i}\cdot (t+1+\ln n)=120 D\cdot \frac{(\ln(2\ell n^2/\delta)+1)^2}{2^{i-1}},\]
  where $D$ is the diameter of $G[C]$. Assuming that $D\leq 2^{i+1}$ (which is implied by the event from property~\ref{p:radii} not holding in the run of the algorithm that constructed $C$), the event above contains the event that $G[C]/\partition_i[C]$ has diameter larger than $480 (\ln(2\ell n^2/\delta)+1)^2$.
 \end{itemize}
 By Lemma~\ref{lem:single-level}, each of the above events happens with probability at most $e^{-t}=\frac{\delta}{2\ell n}$. Since there are at most $2\ell n$ events in question, by the union bound we conclude that $\partition_0,\partition_1,\ldots,\partition_\ell$ is not good with probability at most $\delta$.
\end{proof}

Therefore, the algorithm can simply report failure if $\Ll$ is not good. So from now on we may assume that $\Ll$ is good, and therefore a clustering chain. Properties~\ref{p:hiercl} and~\ref{p:hierdiam} follow directly from the goodness of~$\Ll$, so to finish the proof, it remains to show that $\Ll$ also satisfies property~\ref{p:level}. We do this in the next~claim.

\begin{claim}\label{lem:part:cut2}
For every path $Q$ in $G$ and $i\in \{0,\ldots,\ell-1\}$, the expected number of edges of $Q$ of level $i$ is upper bounded by $2^{-i}\cdot 8\ell(\ln(2\ell n^2/\delta)+1)\cdot \len(Q)$.
\end{claim}
\begin{proof}
Removing edges of levels larger than $i$ breaks $Q$ into subpaths $Q_1,Q_2,\ldots,Q_k$, where each path $Q_j$ is entirely contained in a single cluster of $\partition_{i+1}$. By property~\ref{p:geo} of Lemma~\ref{lem:single-level}, the expected number of edges of level $i$ on each path $Q_j$, $j\in \{1,\ldots,k\}$ is upper bounded by 
\[\frac{\len(Q_j)}{r_i}\cdot 4\ell=2^{-i}\cdot 8\ell(\ln(2\ell n^2/\delta)+1)\cdot \len(Q_j).\]
By the linearity of expectation, the expected total number of edges of level $i$ on $Q$ is upper bounded by
\[2^{-i}\cdot 8\ell(\ln(2\ell n^2/\delta+1)\cdot \sum_{j=1}^k\len(Q_j)\leq 2^{-i}\cdot 8\ell(\ln(2\ell n^2/\delta)+1)\cdot \len(Q).\qedhere\]
\end{proof}

This completes the proof of Lemma~\ref{lem:hierarchical}.

\section{Treewidth tools}\label{sec:tw}

In this section we prepare tools related to treewidth.
We first give a combinatorial lemma about apex-minor-free graphs that can be regarded as a variation on Lemma~\ref{lem:tw:balls}. Let us start with a definition.

For a connected unweighted graph $G$ and integers $a,b,c\geq 0$, an \emph{$(a,b,c)$-contraction sequence} consists~of
\begin{itemize}[nosep]
\item a sequence of graphs $G_0,G_1,\ldots,G_b$, and
\item for every $1 \leq i \leq b$, a collection of pairwise disjoint subgraphs $H_i^1,\ldots,H_i^{a_i}$
of $G_{i-1}$,
\end{itemize}
such that the following conditions hold:
\begin{itemize}[nosep]
\item $G_0 = G$ and $G_b$ is a one-vertex graph.
\item Each $H_i^j$ is of radius at most $c$.
\item $G_i$ is obtained from $G_{i-1}$ by contracting each of the subgraphs $H_i^j$, $1 \leq j \leq a_i$, to a single vertex.
\item It holds that $\sum_{i=1}^b a_i \leq a$.
\end{itemize}
With this definition in place, we can state and prove the promised lemma.

\begin{lemma}\label{lem:tw}
For every apex graph $K$ there exists a constant $\alpha_K$ such that
if a $K$-minor-free graph $G$ admits an $(a,b,c)$-contraction sequence for some integers $a,b,c \geq 1$, 
   then the treewidth of $G$ is bounded by $\alpha_K bc \sqrt{a}$. 
\end{lemma}
\begin{proof}
We let $\alpha_K$ be the constant provided by Lemma~\ref{lem:tw:balls}.

For $1 \leq i \leq b$ and $1 \leq j \leq a_i$, let $v_i^j$ be a vertex of $H_i^j$ such that all vertices of $H_i^j$
are within distance~$c$ from $v_i^j$ in $H_i^j$. 
When $H_i^j$ is contracted onto a single vertex in the construction of $G_i$ from~$G_{i-1}$, we imagine that $H_i^j$ is contracted onto $v_i^j$ and we keep the name
$v_i^j$ for the obtained vertex. Thus, we have $V(G_i) = \{v_i^j~\colon~1 \leq i \leq a_i\} \cup (V(G_{i-1}) \setminus \bigcup_{j=1}^{a_i} V(H_i^j))$. In particular, $V(G_i) \subseteq V(G)$ for every~$0 \leq i \leq b$. 

Let $Z = \{v_i^j~\colon~1 \leq i \leq b, 1 \leq j \leq a_i\}$;
thus $Z$ is a subset of $V(G)$ of size at most $a$.
The crucial observation is as follows: for every $1 \leq i \leq b$ and $v \in V(G_i)$, we
have
\begin{equation}\label{eq:tw:radius}
\mathrm{dist}_{G_i}(v, Z \cap V(G_i)) \geq \mathrm{dist}_{G_{i-1}}(v, Z \cap V(G_{i-1})) - c.
\end{equation}
Indeed, consider a shortest path $Q$ connecting $v$ with a vertex of $Z \cap V(G_i)$ in $G_i$.
All vertices and edges of $Q$, except for possibly the last edge and the endpoint $z \in Z \cap V(G_i)$, exist in $G_{i-1}$. 
Thus, $Q$ either is entirely contained in $G_{i-1}$, or it can be extended to a path from $v$ to $Z \cap V(G_{i-1})$ in $G_{i-1}$ 
by appending a path of length at most $c$ within the graph $H_i^j$ to which $z$ belongs. This proves~\eqref{eq:tw:radius}.

As $G_b$ consists of a single vertex, from~\eqref{eq:tw:radius} we obtain
that every vertex of $G_0=G$ is within distance at most $bc$ from a vertex of $Z$. As $|Z| \leq a$, 
the lemma follows from Lemma~\ref{lem:tw:balls}.
\end{proof}

Next, we use Lemma~\ref{lem:tw} to prove a statement that is the cornerstone of our approach. Intuitively it says that in a clustering chain of an edge-weighted apex-minor-free graph, one can find many disjoint balanced separators consisting of clusters. Again, we need a few definitions.

For a (possibly edge-weighted) graph $G$, a {\em{cut}} is a family of pairwise disjoint vertex subsets of $G$, each inducing a connected subgraph of $G$.
For a cut $\Aa$, let $F(\Aa)$ be the set of those edges $e$ of $G$ for which there exists $A \in \Aa$
that contains exactly one endpoint of $e$.
Note that since every part $A \in \Aa$ induces a connected subgraph of $G$, $G[A]$ is a connected component of $G-F(\Aa)$. We say that a cut $\Aa$ is \emph{balanced} if the vertex set of every component of $G-F(\Aa)$ is contained in $\Aa$
or consists of at most $|V(G)|/2$ vertices. A pair of cuts $\Aa,\Aa'$ is {\em{non-conflicting}} if every set in $\Aa\cap \Aa'$ consists of one vertex.

Let $\Ll=(\partition_0,\partition_1,\ldots,\partition_k)$ be a clustering chain of a connected graph $G$. We say that a cut $\Aa$ {\em{respects}}~$\Ll$ if $\Aa\subseteq \bigcup_{i=0}^k \partition_i$; in other words, every part $A\in \Aa$ is a cluster belonging to one of partitions $\partition_0,\partition_1,\ldots,\partition_k$. A {\em{cut packing}} in $\Ll$ is a family of pairwise non-conflicting cuts that respect $\Ll$.

\begin{lemma}\label{lem:emb:crucial}
For every apex graph $K$ there exists a constant $\gamma_K>0$ such that the following holds.
Let $G$ be a $K$-minor-free graph and let $\Ll=(\partition_0,\partition_1,\ldots,\partition_k)$ be a clustering chain of $G$ such that for every $i\in \{0,1,\ldots,k-1\}$ and $C\in \partition_{i+1}$, the quotient graph $G[C]/\partition_i[C]$ has radius at most $c$. Suppose $\Ff$ is a cut packing in $\Ll$ such that every cut in $\Ff$ has size at most $\tau$ and
\[ |\Ff|\leq \gamma_K\cdot \frac{\tau}{k^2c^2}.\]
Then there exists a balanced cut $\Aa$ respecting $\Ll$ and of size at most $\tau$ such that $\Ff\cup \{\Aa\}$ is also a cut packing in $\Ll$. Moreover, given $G$, $\Ll$, and $\Ff$, such a balanced cut $\Aa$ can be found in polynomial time.
\end{lemma}
\begin{proof}
Let $\wh{\partition}\coloneqq \bigcup_{i=0}^k \partition_i$; the elements of $\wh{\partition}$ will be called {\em{clusters}}. Call a cluster $D\in \wh{\partition}$ {\em{free}} if $D$ consists of one vertex or is not contained in any member of $\Ff$. If $D$ is not free, it is called {\em{used}}.
A free cluster $D$ is \emph{maximal} if there is no cluster $D'\supseteq D$ that is free. Let $\Dd$ be the family consisting of all maximal free clusters. Note that $\Dd$ is a partition of $V(G)$. Let $H=G/\Dd$. In other words, $H$ is obtained from $G$ by contracting every maximal free cluster into a single vertex.

Let $\xi\coloneqq |\Ff|$. The crucial observation is as the following: $H$ admits a $(\tau\xi,k,c)$-contraction sequence. Indeed, we can proceed through levels $i=1,2,\ldots,k$, and upon considering level $i$, contract all used clusters contained in $\partition_i$ that have not been contracted so far. Since there are at most $\tau\xi$ used clusters in total, and at every point we contract a graph of radius at most $c$ (due to the assumption on the quotient graphs appearing in the clustering chain $\Ll$), this is a $(\tau\xi,k,c)$-contraction sequence.

By Lemma~\ref{lem:tw}, the treewidth of $H$ is bounded by $\alpha_K\cdot \sqrt{\tau\xi} \cdot kc$ for a constant
$\alpha_K>0$ depending only on $K$. Using the algorithm of Feige et al.~\cite[Theorem~6.4]{FeigeHL08}, in polynomial time we can compute a tree decomposition $(T,\bag)$ of $G$ where every bag has size at most $\beta_K\cdot \sqrt{\tau\xi}\cdot kc$, where $\beta_K>0$ is again a constant depending only on $K$.

We now use the following standard statement about tree decompositions, see e.g.~\cite[Lemma 7.19]{platypus}.

\begin{claim}\label{cl:balbag} Suppose $(T,\bag)$ is a tree decomposition of a graph $G$ and $\omega\colon V(G)\to \R_{\geq 0}$ is a nonnegative weight function on the vertices of $G$. Then there exists a node $x$ of $T$ such that for every connected component $J$ of $G-\bag(x)$, we have $\omega(V(J))\leq \omega(V(G))/2$.
\end{claim}

Setting $\omega(D)\coloneqq |D|$ for each $D\in \Dd$, we find that there exists a node $x$ of $T$ such that every connected component of $H-\bag(x)$ consists of clusters of total size at most $|V(G)|/2$. Note that such a node $x$ can be found in polynomial time by examining the nodes of $T$ one by one. Denote $\Aa\coloneqq \bag(x)$.

As $\Aa\subseteq V(H)=\Dd\subseteq \wh{\partition}$, $\Aa$ is a cut that respects $\Ll$.
Since $\Aa$ consists only of free clusters, it is also non-conflicting with every member of $\Ff$. Therefore, $\Ff\cup \{\Aa\}$ is a cut packing in $\Ll$. Further, the property of $\Aa$ described in the previous paragraph ensures that $\Aa$ is balanced. Finally, observe that
\[|\Aa|\leq \beta_K\cdot \sqrt{\tau\xi}\cdot kc\leq \frac{\beta_K}{\sqrt{\gamma_K}}\cdot \tau.\]
Hence, if we set $\gamma_K\coloneqq \beta_K^2$, then we have $|\Aa|\leq \tau$ and $\Aa$ satisfies all the required properties.
\end{proof}

\newcommand{\embed}{\mathtt{embed}}
\newcommand{\findsplit}{\mathtt{split}}

\newcommand{\totell}{{\hat{\ell}}}

\section{Embedding} \label{sec:embed}

In this section we prove Theorem~\ref{thm:main} for the case of apex-minor-free graphs.
Let then $G$ be the input edge-weighted graph that is $K$-minor-free, where $K$ is a fixed apex graph. We may assume that $G$ is connected, for otherwise we treat every connected component of $G$ separately.
We may also assume that for every $uv \in E(G)$, we have $\len(uv) = \dist_G(u,v)$.  
Let $n$ be the vertex count of $G$
and let $\totell$ be the least integer such that the ratio of the sum of the weights of all edges of $G$ to the smallest
edge weight of $G$ is strictly less than $2^\totell$.
Note that $\totell = \Oh_K(\ell + \ln n)$ where $\ell$ is the logarithm of the stretch of the metric of~$G$.
By rescaling, we assume that every distance in $G$ is larger than one and the total length of all edges of $G$
is at most $2^\totell$.
We have that any simple path in $G$ has length at most $2^\totell$, hence, in particular,
  any connected subgraph of $G$ has diameter at most $2^\totell$.

Let us also fix the accuracy parameter $\eps>0$ given on input.

\paragraph*{Procedure $\embed$.}
We first describe a recursive embedding procedure $\embed$. The input to the procedure is a connected induced subgraph $G'$ of $G$. The output is a metric embedding $(H',\eta')$ of $G'$ together with an elimination forest $F'$ of $H'$. Procedure $\embed(G')$ proceeds as follows:
\begin{itemize}[nosep]
 \item If $G'$ has one vertex, return the trivial embedding: $H'=G'$ and $\eta'$ is the identity on~$V(G')$.
 \item Otherwise, apply a subprocedure $\findsplit(G')$, which we will describe later. This procedure finds a subset of edges $F\subseteq E(G)$ (called \emph{cutedges}) and a subset of vertices $Z \subseteq V(G)$ (called \emph{portals}).
 \item Let $\Dd$ be the set of connected components of the graph $G'-F$. For every component $D\in \Dd$, recurse on $D$ to obtain a metric embedding $(H_D,\eta_D)$ of $D$ and an elimination forest $F_D$ of $H_D$.
 \item Construct an embedding $(H',\eta')$ of $G'$ as follows. First, $H'$ is obtained by taking the disjoint union of graphs $\{H_D\colon D\in \Dd\}$ and adding, for every portal $z\in Z$, a copy $z'$ of $z$; this copy is made adjacent to every vertex $u$ of $G'$ by an edge of length $\dist_G(z,u)$. Next, we set $\eta'\coloneqq \bigcup_{D\in \Dd} \eta_D$. 
 \item Finally, construct an elimination forest $F'$ of $H'$ by first taking the disjoint union of forests $\{F_D\colon D\in \Dd\}$, and then iteratively adding all vertices $\{z'\colon z\in Z\}$ in any order. Each time $z'$ is added to $F'$, we make it the new root and attach all former roots as children of $z'$.
\end{itemize}
A straightforward induction shows that procedure $\embed$ indeed outputs a metric embedding of the graph provided on input.
Hence, it remains to specify the procedure $\findsplit$. This will be done so that the following two conditions are satisfied:
\begin{itemize}[nosep]
 \item The depth of the recursion tree of $\embed$ is bounded by $\totell \cdot \lceil \log_2 n \rceil = \Oh((\ell + \ln n) \ln n)$.
 \item In every recursion call, $\findsplit$ returns a portal set $Z$ of size polynomial in $\totell$, $1/\eps$, and~$\ln n$.
\end{itemize}
Thus, the depth of the eventually constructed elimination forest of $G$ will be polynomial in $\totell$, $1/\eps$, and $\ln n$, which will certify the treedepth of the target graph in the obtained metric embedding of $G$.

\paragraph*{Procedure $\findsplit$.} We now describe the procedure $\findsplit$. Recall that the input to this procedure is a connected subgraph $G'$ of $G$, where we may assume that $G'$ has at least two vertices. 
It follows that~$D$, the diameter of $G'$, is strictly larger than $1$. Let then $\ell'\in \{1,\ldots,\totell\}$ be such that $2^{\ell'-1}<D\leq 2^{\ell'}$. We call $\ell'$ the {\em{level}} of the subgraph $G'$.
Let
\begin{align*}
& \delta\coloneqq \frac{\eps}{c\totell n\ln^2 n},& \qquad & \xi\coloneqq 64\totell^3\lceil\log_2 n\rceil\left(\ln(2\totell n^2/\delta)+1\right)\cdot (1/\eps), \\ & \sigma\coloneqq 480(\ln(2\totell n^2/\delta)+1)^2, & \qquad & \tau\coloneqq (\xi+1)\cdot \gamma_K \totell^2 \sigma^2, 
\end{align*}
where $c>0$ is a constant that will be determined later and $\gamma_K$ is the constant provided by Lemma~\ref{lem:emb:crucial}.
Note that
\[\tau \in \Oh_K\left(\totell^5/\eps\cdot \ln n\cdot (\ln n+\ln \totell+\ln (1/\eps))^5\right).\]

We apply the algorithm of Lemma~\ref{lem:hierarchical:stmt} to $G'$ with parameter $\delta$, thus obtaining a clustering chain $\Ll=(\partition_0,\partition_1,\ldots,\partition_{\ell'})$ of~$G'$. It may happen that this algorithm reports failure; this happens with probability at most $\delta$. In this case, we terminate the whole effort, and instead of trying to compute a metric embedding of $G$ ourselves, we apply the algorithm of Fakcharoenphol, Rao, and Talwar~\cite{FRT04} that finds an embedding of $G$ into a tree (thus, a graph of treedepth $\Oh(\ln n)$) with expected multiplicative distortion $\Oh(\ln n)$.
From now on, we suppose the algorithm of Lemma~\ref{lem:hierarchical:stmt} succeeded in finding the clustering chain $\Ll$.

Note that by Lemma~\ref{lem:hierarchical:stmt}, property~\ref{p:stmt:hierdiam}, for every $i\in \{0,\ldots,\ell'-1\}$ and cluster $C\in \partition_{i+1}$, the diameter of the quotient graph $G[C]/\partition_i[C]$ is bounded by $\sigma$. So
starting with $\Ff=\emptyset$, we may apply the algorithm of Lemma~\ref{lem:emb:crucial} repeatedly to find a cut packing $\Ff$ in $\Ll$ consisting of balanced cuts of size at most $\tau$ so that
\[|\Ff|\geq \frac{\tau}{\gamma_K \totell^2 \sigma^2}=\xi+1.\]
Once $\Ff$ is constructed, we discard from $\Ff$ the cut $\{V(G')\}$, if present; let $\Ff'$ be the obtained packing of size at least $\xi$. Then we pick a cut $\Aa\in \Ff'$ uniformly at random and define the output of $\findsplit(G')$ as~follows:
\begin{itemize}[nosep]
 \item as cutedges we set $F\coloneqq F(\Aa)$; and
 \item as portals we set a set $Z$ consisting of one vertex from each cluster $A\in \Aa$, selected arbitrarily.
\end{itemize}
This concludes the description of procedure $\findsplit$.

\paragraph*{Properties.} We first observe that in procedure $\embed$, at every recursion call we observe progress in significantly reducing either the diameter or the vertex count.

\begin{claim}\label{cl:progress}
 Suppose a call $\embed(G')$ invokes a subcall of $\embed(G'')$, for some induced subgraph $G''$ of~$G'$. Then $|V(G'')|\leq |V(G')|/2$ or the level of $G''$ is strictly smaller than the level of $G'$. 
\end{claim}
\begin{proof}
 Let $\ell'$ be the level of $G'$.
 Recall that procedure $\embed(G')$ calls $\findsplit(G')$ to find suitable cutedges $F$ and portals $Z$ to recurse. Let $\Aa$ be the cut used by $\findsplit(G')$ to find define $F$ and $Z$. Recall that $\Aa$ is balanced, which means that every connected component of $G'-F(\Aa)$ that is not a member of $\Aa$ has at most $|V(G')|/2$ vertices. On the other hand, if $\Ll=(\partition_0,\partition_1,\ldots,\partition_{\ell'})$ is the hierarchical clustering considered in procedure $\findsplit(G')$, then $\Aa\subseteq \partition_0\cup \partition_1\cup \ldots\cup \partition_{\ell'-1}$, because the cut $\partition_{\ell'}=\{V(G')\}$ was explicitly excluded from sampling. Since every cluster within partitions $\partition_0,\partition_1,\ldots,\partition_{\ell'-1}$ induces a subgraph of has diameter at most $2^{\ell'-1}$ (Lemma~\ref{lem:hierarchical:stmt}, property~\ref{p:stmt:hiercl}), it follows that for each $A\in \Aa$, $G[A]$ has level strictly smaller than $\ell'$. The claim follows.
\end{proof}

An immediate consequence of Claim~\ref{cl:progress} is the following.

\begin{claim}\label{cl:rec-depth}
 The depth of the recursion tree of $\embed(G)$ is bounded by $\totell\lceil \log_2 n\rceil$.
\end{claim}
\begin{proof}
 Consider any root-to-leaf path $P$ in the recursion tree of $\embed(G)$. By Claim~\ref{cl:progress}, we may color every edge of $P$, say with parent call $\embed(G')$ and child call $\embed(G'')$, blue if $|V(G'')|\leq |V(G')|/2$ and red if the level of $G''$ is strictly smaller than the level of $G'$. Clearly, $P$ cannot contain more than $\totell$ red edges in a row, and there are no more edges on $P$ after the $\lceil\log_2 n\rceil$th blue edge. The claim~follows.
\end{proof}

Observe that every call of procedure $\findsplit$ returns a portal set $Z$ of size at most $\tau$. Hence, from Claim~\ref{cl:rec-depth} it follows that if $\embed(G)$ returns a metric embedding $(H,\eta)$ and an elimination forest $F$ of~$H$, then the depth of $F$ is bounded by
\[1+\tau\totell\lceil\log_2 n\rceil=\Oh_K\left((\ell + \ln n)^6/\eps\cdot \ln^2 n\cdot (\ln n+\ln \ell+\ln (1/\eps))^5\right).\]
Consequently, the above expression is also an upper bound on the treedepth of $H$, as desired.

\medskip

It remains to bound the expected distortion.

Consider first the corner case when one of the calls of procedure $\findsplit$ failed. Then we invoked the algorithm of~\cite{FRT04}, which computed a metric embedding into a tree with expected multiplicative distortion $\Oh(\ln n)$. Every invocation of procedure $\findsplit$ fails with probability at most $\delta$, and by Claim~\ref{cl:rec-depth}, the total number of calls to $\findsplit$ is bounded by $n\totell\lceil\log_2 n\rceil$. Therefore, the total contribution to the expected multiplicative distortion from this case is
\[\delta\cdot n\totell\lceil\log_2 n\rceil\cdot \Oh(\ln n),\]
which is upper bounded by $\eps/2$ assuming we choose $c$ large enough.

From now on we assume that all calls to $\findsplit$ succeeded and our algorithm has indeed constructed a metric embedding $(H,\eta)$ of $G$.
Fix a pair of vertices $u_1,u_2\in V(G)$ and let $Q$ be a shortest path connecting $u_1$ and $u_2$ in $G$.
Note that within the recursion tree of $\embed(G)$,
there exists a unique call $\embed(G')$ such that the whole path $Q$ is contained in $G'$, but this cannot be said about any call $\embed(G'')$ invoked within $\embed(G')$.
This means that $Q$ contains at least one of the cutedges belonging to $F(\Aa)$, where $\Aa$ is the cut considered in the call $\findsplit(G')$. Equivalently, there exists a cluster $A\in \Aa$ such that $Q$ contains an edge with exactly one endpoint in $A$. 
Let $z$ be the vertex of $A$ chosen to the portal set $Z$ by $\findsplit(G')$, and let $z'\in V(H)$ be the copy of $z$ created within $\embed(G')$.
By triangle inequality, we have
\[\dist_H(\eta(u_1),\eta(u_2))\leq \dist_H(\eta(u_1),z')+\dist_H(\eta(u_2),z')=\dist_G(u_1,z)+\dist_G(u_2,z).\]
Since $Q$ intersects $A$, say at vertex $q$, we have
\[\dist_G(u_1,z)\leq \dist_G(u_1,q)+\dist_G(q,z)\leq \dist_G(u_1,q)+D,
\]
where $D$ is the diameter of $G[A]$. Similarly $\dist_G(u_2,z)\leq \dist_G(u_2,q)+D$, implying that
\begin{align*}\dist_H(\eta(u_1),\eta(u_2))& \leq \dist_G(u_1,z)+\dist_G(u_2,z) \leq \dist_G(u_1,q)+\dist_G(u_2,q)+2D\\
& =\len(Q)+2D = \dist_G(u_1,u_2)+2D.
 \end{align*}
 
For an index $h\in \{0,1,\ldots,\totell\lceil \log_2 n\rceil-1\}$, we introduce a random variable $X_h$ defined as follows:
\begin{itemize}
 \item If at depth $h$ of the recursion tree of $\embed(G)$ there is no call $\embed(G')$ such that $Q$ is entirely contained in $G'$, then $X_h=0$.
 \item Otherwise, consider the unique call $\embed(G')$ at depth $h$ such that $Q$ is entirely contained in $G'$. Then we let
 \begin{equation}\label{eq:wydra}
X_h\coloneqq \sum_{A\in \Aa} 2\mathrm{diam}(G[A])\cdot |E(Q)\cap F(A)|,  
 \end{equation}
 where $\Aa$ is the cut considered in procedure $\findsplit(G')$, $F(A)$ is the set of edges with exactly one endpoint in $A$, and $\mathrm{diam}(G[A])$ is the diameter of $G[A]$.
\end{itemize}
The reasoning of the previous paragraph together with Claim~\ref{cl:rec-depth} proves that
\[\dist_H(\eta(u_1),\eta(u_2))\leq \dist_G(u_1,u_2)+\sum_{h=0}^{\totell\lceil \log_2 n\rceil-1} X_h.\]

Consider then a fixed $h\in \{0,1,\ldots,\totell\lceil \log_2 n\rceil-1\}$ such that there is a call $\embed(G')$ at depth $h$ such that $Q$ is entirely contained in $G'$. Recall that within $\findsplit(G')$ we constructed a clustering chain $\Ll=(\partition_0,\partition_1,\ldots,\partition_{\ell'})$ and a cut packing $\Ff'$, and we sampled a balanced cut $\Aa$ from $\Ff'$. Since members of $\Ff'$ are pairwise non-conflicting, for every fixed cluster $C\in \bigcup_{i=1}^{\ell'-1} \partition_i$ of positive diameter, the probability that $C$ belongs to $\Aa$ is upper bounded by $\frac{1}{|\Ff'|}\leq \frac{1}{\xi}$. Moreover, if $C\in \partition_i$, then the diameter of $G[A]$ is upper bounded by $2^i$ (Lemma~\ref{lem:hierarchical:stmt}, property~\ref{p:stmt:hiercl}).
By property~\ref{p:stmt:level} of Lemma~\ref{lem:hierarchical:stmt}, we infer the following:
\begin{align*}
\Exp X_h &\leq \Exp\left(\sum_{i=1}^{\ell'-1} \ \ \sum_{C \in \partition_i} |E(Q) \cap F(C)| \cdot [C \in \Aa] \cdot 2^{i+1}\right) \\
    &\leq  \frac{1}{\xi}\cdot \sum_{i=1}^{\ell'-1} 2^{i+2}\cdot \Exp\left(\sum_{e\in E(Q)} [e\textrm{ is of level }i\textrm{ in }\Ll]\right)\\
    &\leq  \frac{1}{\xi}\cdot \sum_{i=1}^{\ell'-1} 2^{i+2}\cdot 2^{-i}\cdot 8\totell(\ln(2\totell n^2/\delta)+1)\cdot \len(Q)\\
    &\leq \frac{1}{\xi}\cdot 32\ell^2(\ln(2\totell n^2/\delta)+1)\cdot \len(Q)\\
    &= \frac{\eps}{2\totell\lceil\log_2 n\rceil}\cdot \len(Q).
\end{align*}

By~\eqref{eq:wydra}, we conclude that
\[\Exp\ \dist_H(\eta(u_1),\eta(u_2))\leq (1+\eps/2)\dist_G(u_1,u_2).\]
Taking into account the additional summand of $(\eps/2)\cdot \dist_G(u_1,u_2)$ from the case when one of the calls to $\findsplit$ reports a failure, we conclude that the expected multiplicative distortion of the embedding is $(1+\eps)$, as desired. This concludes the proof of Theorem~\ref{thm:main}.

\section{Extension to proper minor-closed graph classes}\label{sec:minor}

\newcommand{\torso}{\mathsf{torso}}
\newcommand{\adh}{\mathsf{adh}}
\newcommand{\apices}{\mathsf{apices}}
\newcommand{\Capices}{c_{\apices}}

In this section we extend the main result to $K$-minor-free graphs for a fixed graph $K$
that is not necessarily an apex graph. 
The arguments of Section~\ref{sec:tw} crucially rely on the graph class excluding an apex graph. 
Here, we desing a workaround using the structure theorem of Robertson and Seymour~\cite{gm16}.

\paragraph{Tools from the graph minors theory.}
We need a bit more notation on tree decompositions. Let $(T,\bag)$ be a tree decomposition of a graph $G$.
For an edge $st \in E(T)$, the \emph{adhesion} of $st$, denoted $\adh(st)$, equals $\bag(s) \cap \bag(t)$.
For a vertex $t \in V(T)$, the \emph{torso} of $t$, denoted $\torso(t)$, is the graph obtained from
$G[\bag(t)]$ by turning the adhesion $\adh(st)$ into a clique, for every $s \in N_T(t)$. 
If $G$ is an edge-weighted graph, every new edge $uv$ in $\adh(st)$ is assigned length equal to $\dist_G(u,v)$.

We need the following corollary of the structure theorem for graphs excluding a fixed minor. 
Unfortunately, we did not find a similar form in the literature; we discuss how to obtain this corollary from
existing results in Appendix~\ref{app:minor}.
\begin{theorem}\label{thm:structure}
For every graph $K_0$ there exists a constant $\Capices$, an apex graph $K$, and a polynomial-time algorithm
that, given a $K_0$-minor-free graph $G$, computes a tree decomposition $(T,\bag)$ of $G$
and for every $t \in V(T)$ a set $\apices(t) \subseteq \bag(t)$ of size at most $\Capices$
so that for every $t \in V(T)$, the graph $\torso(t)-\apices(t)$ does not contain $K$ as a minor.
\end{theorem}
Note that in the decomposition $(T,\bag)$ returned by Theorem~\ref{thm:structure}, every adhesion
needs to be of size strictly smaller than $\Capices + |V(K)|$: an adhesion is turned into a clique in the torso
while the clique number of $\torso(t)-\apices(t)$ is strictly less than $|V(K)|$ as this graph is $K$-minor-free.

We also need the following observation (proved formally in Appendix~\ref{app:minor}).
\begin{lemma}\label{lem:add-simplicial}
For every apex graph $K$ there exists an apex graph $K'$ such that if $G$
is $K$-minor-free and $G'$ is created from $G$ by adding, for every nonempty clique $A$ in $G$,
a new vertex $v_A$ with neighborhood $N_{G'}(v_A) = A$, then $G'$ is $K'$-minor-free.
\end{lemma}

\paragraph{Initial steps.}
Armed with Theorem~\ref{thm:structure}, we now describe modifications to the $\findsplit$ procedure.
Recall that the input is an induced subgraph $G'$ of the input graph $G$
that is connected and has at least two vertices. 
We denote by $D$ the diameter of $G'$ and by $\ell' \in \{1,\ldots,\totell\}$ the least integer such that $D \leq 2^{\ell'}$.
Recall also that in $G$, we assume that for every $uv \in E(G)$, the length of $uv$ equals $\dist_G(u,v)$.

Instead of applying Lemma~\ref{lem:hierarchical:stmt} directly to $G'$, we proceed as follows.
We apply the algorithm of Theorem~\ref{thm:structure} to $G'$, obtaining a tree decomposition $(T,\bag)$
and subsets $\apices(t)$ for all $t \in V(T)$. 
We find a node $t \in V(T)$ such that every connected component of $G'-\bag(t)$ has at most $|V(G')|/2$ vertices (see Claim~\ref{cl:balbag}).
We denote $G^1 \coloneqq G' - \apices(t)$ and $G^2 \coloneqq \torso(t) - \apices(t)$.
We designate $\apices(t)$ as portals and delete them from $G'$; thus, we will focus on $G^1$ and $G^2$ in what follows.

\paragraph{Modifying the clustering step.}
We apply the clustering step (Lemma~\ref{lem:hierarchical:stmt}) to $G^1$ with the following modification:
in the tie-breaking order $\preceq$, we put all vertices of $V(G^2)$ before the remaining vertices of $V(G^1) \setminus V(G^2)$. 
In other words, when the algorithm chooses the next center of a cluster, we prefer a center lying in $\bag(t) \setminus \apices(t)$.
We obtain a clustering chain $\Ll^1 = (\partition^1_0,\partition^1_1,\ldots,\partition^1_{\ell'})$.

Recall that in $\torso(t)$, each new edge put within an adhesion has length equal to the distance
of the endpoints in $G'$. Hence, both in $G^1$ and $G^2$ we have the property that
the length of an edge equals the distance between its endpoints. 
This implies that, for every cluster $C$ in $\Ll$, if $C$ intersects $V(G^2)$, then $G^2[C \cap V(G^2)]$ is a connected
subgraph of $G^2$ of diameter not larger than the diameter of $G^1[C]$. 
Hence, we can obtain from $\Ll^1$ a clustering chain $\Ll^2 = (\partition^2_0,\ldots,\partition^2_{\ell'})$ of $G^2$
by restricting every cluster $C$ to $C \cap V(G^2)$ and deleting clusters disjoint with $V(G^2)$. 
Observe that, because adhesions are turned into cliques in the torso, for every $1 \leq i < \ell'$
and $C \in \partition^1_{i+1}$ with $C \cap V(G^2) \neq \emptyset$, 
the diameter of $G^2[C \cap V(G^2)]/\partition^2_i[C \cap V(G^2)]$ is not larger
than the diameter of $G^1[C]/\partition^1_i[C]$.

Furthemore, we make the following crucial observation: the obtained distribution of clustering chain $\Ll^2$ of $G^2$
is the same as if we just applied the algorithm of Lemma~\ref{lem:hierarchical:stmt} directly to $G^2$ without any 
modifications (and with the same order $\preceq$, but restricted to the vertices of $G^2$). 

\paragraph{Obtaining a cut packing.}
Having obtained $\Ll^1$ and $\Ll^2$, we proceed as follows.
Let $G^3$ be an unweighted graph obtained from $G^2$ by adding, for every $s \in N_T(t)$, a vertex $v_s$ adjacent to all vertices of $\adh(st)$.
Theorem~\ref{thm:structure} implies that $G^2$ is $K$-minor-free for an apex graph $K$ that depends only on $K_0$.
Now, Lemma~\ref{lem:add-simplicial} implies that $G^3$ is $K'$-minor-free for another apex graph $K'$ that depends only on $K$.
We extend $\Ll^3$ to a clustering chain of $G^3$ by adding at every level singleton clusters $(\{v_s\})_{s \in N_T(t)}$,
except for the top level where we now have $V(G^3)$ as the only cluster.
Note that the diameter of every quotient graph discussed in property~\ref{p:hierdiam} increased by at most $2$. This increase is neglible and will be be promptly ignored in the argumentation to follow.

We use the graph $G^3$ and the clustering chain $\Ll^3$ in the repeated application of Lemma~\ref{lem:emb:crucial} to
obtain a large cut packing $\Ff$.
We do one modification to the process: in the definition of a balanced separator, we put weight $1$ on every vertex
of $V(G^2)$ but, for every $s \in N_T(t)$, the weight of $v_s$ is defined as $|V_s|$, where
$V_s$ is the set of vertices
of $V(G^1) \setminus V(G^2)$ that belong to any bag of  the connected component of $T-\{st\}$ containing $s$.
Hence, the total weight of all vertices of $G^3$ is $|V(G^1)|$ and 
every $\Aa \in \Ff$ has the property that every connected component of $G^3-F(\Aa)$ is either an element of $\Aa$
or has weight at most $|V(G^1)|/2$.

To obtain $\Ff'$ from $\Ff$, we not only discard the cut $\{V(G^3)\}$ if present, but also,
for every $\Aa \in \Ff$, whenever $\Aa$ contains a cluster $\{v_s\}$ for some $s \in N_T(t)$,
we replace it within $\Aa$ with single-vertex clusters $\{\{v\}~\colon~v \in \adh(st)\}$. 
Since we replace $\{v_s\}$ with single-vertex clusters, the members of the cut packing $\Ff'$ are still pairwise non-conflicting.
Since every adhesion is of size less than $\Capices + |V(K)|$, we increase the size of every element of $\Ff'$ by at most a constant
multiplicative factor.
Since the weight of every $v_s$ is at most $|V(G^1)|/2$, every cut in $\Ff'$ is still balanced. 
We obtain that every cut $\Aa \in \Ff'$ contains only clusters of $\Ll^2$.

\paragraph{Cutting.}
The cut packing $\Ff'$ in clustering chain $\Ll^2$ projects to a cut packing $\Ff^1$
in clustering chain $\Ll^1$ as follows: for every $\Aa \in \Ff'$ and $C\in \Aa$, we put into $\Aa^1$ the unique inclusion-wise minimal cluster $C'$ appearing in $\Ll^1$ such that $C'\cap V(G^2)=C$, and $\Ff^1$ consists of all cuts $\Aa^1$ constructed in this manner. Note that thus, if $C$ is a single-vertex cluster, then $C'=C$.

We make two observations. First, it follows directly from the construction above and the fact that the members of $\Ff'$ are pairwise non-conflicting that the members of
$\Ff^1$ are also non-conflicting.
Second, thanks to the weights used in~$G^3$, the cuts in $\Ff^1$ are balanced in the usual way in the graph $G^1$:
for every $\Aa^1 \in \Ff^1$, every component of $G^1-F(\Aa^1)$ is either a member of $\Aa^1$ or contains at most $|V(G^1)|/2$ vertices.
To see this, note that for every component $D^1$ of $G^1-F(\Aa^1)$ that is not member of $\Aa^1$,
 if $\Aa$ is the member of $\Ff'$ corresponding to~$\Aa^1$, then there is a component $D'$ of $G^3-F(\Aa)$
 that is not a member of $\Aa$,
 $D' \cap V(G^2) = D^1 \cap V(G^2)$, and for every $s \in N_T(t)$ we have 
 that $D^1 \cap V_s \neq \emptyset$ implies $v_s \in D'$. 
 Then, $|D^1|$ is bounded by the weight of $D'$, which in turn is bounded by $|V(G^1)|/2$
 as every cut in $\Ff'$ is balanced.

 Hence, we can proceed with $G^1$ and $\Ff^1$ as in the original algorithm: choose $\Aa^1 \in \Ff^1$ at random
 and cut along it.
 
 \medskip
 
 Since we obtained a family of pairwise non-conflicting  balanced cuts in $G^1$ in the end, the analysis of the depth 
 of the recursion and the expected distortion is the same. 
 The depth of the obtained elimination forest of the host graph grew by a constant factor:
 we add a constant number of vertices of $\apices(t)$ in the beginning and then,
 in the process of constructing $\Ff'$ from $\Ff$, we may have increased the size of cuts by a constant factor.
 This finishes the proof of the extension to $K$-minor-free graphs, for an arbitrary fixed graph $K$.

\newcommand{\OPT}{\mathsf{OPT}}

\section{Algorithmic applications}\label{sec:alg-ap}

In this section, we present some algorithmic applications of our embedding in \Cref{thm:main}, specifically, in solving the CVRP and  capacitated clustering/facility location problems. 

The algorithms for all of these problems follow the same framework. Suppose that we are given an instance $\Pi_G$ of a problem $\Pi$ to solve on an input planar graph $G$, and $\OPT_G$ be the optimal solution of $\Pi_G$. We use  $|\OPT_G|$ to denote the value of $\OPT_G$.   Our algorithm for solving $\Pi_G$ has three steps.

\begin{itemize}
	\item \textbf{Step 1.~} Normalize the edge weights of $G$ so that $w(e) \in [1,n^{\Oh(1)}]$ for every edge $e\in E(G)$. Then embed $G$ into a graph  $H$ of treewidth $k_H = \polylog(n)$ by \Cref{thm:main} via embedding $\eta$.  
	\item \textbf{Step 2.~} Let $\Pi_H$ be the instance of problem $\Pi$ to solve on graph $H$. We  solve $\Pi_H$ in $H$ to get an optimal solution $\OPT_H$ of $\Pi_H$. 
	
	\item \textbf{Step 3.~} Lift $\OPT_H$ to a solution $\widehat{\OPT}_G$ of the problem $\Pi_G$ in the input graph $G$.  %
\end{itemize}

We note that if we skip the normalization in Step 1, our algorithm will have a factor of $2^{\poly(\ell)}$ in the running time.  We show in all of the problems we consider below, we can do the normalization step, and  the framework gives a $(1+\eps)$ approximation algorithm.

\paragraph{Capacitated vehicle routing problem.~} Here we show the implementation of each step for the CVRP. In Step 1, we first compute a constant approximation solution, denoted by $A$, for the CVRP problem in polynomial time~\cite{HK85}. (Indeed, even any  $\poly(n)$-approximation suffices for our purpose.).  Then we round every edge of weight at most $\eps |A|/(\alpha n^{3})$ to $\eps |A|/(\alpha n^{3})$ for some big constant $\alpha$, and  remove every edge of weight  more than $|A|$ out of the graph. Then scale every edge so that the minimum weight edge is $1$. The maximum weight edge will be $\Oh(n^3/\eps)$.  In Step 2, we use the algorithm for small treewidth graphs by Jayaprakash and Salavatipour~\cite{JS23}:

\begin{theorem}[Jayaprakash and Salavatipour~\cite{JS23}]\label{thm:JS23} Given a parameter $\eps > 0$ and  graph $G$ of treewidth at most $k$, one can find a $(1+\epsilon)$-approximate solution for :
	\begin{itemize}
		\item the CVRP with \emph{unit demands} in time $n^{\Oh(k^2\log^3n/\eps)}$ 
		\item the CVRP with \emph{splittable/unsplittable demands } in time $n^{\Oh(k^2 \log^{2c+3}n/\eps^2)}$ when the capacity of the vehicle is $Q = n^{\Oh(\log^c n)}$ for some constant $c > 0$.
	\end{itemize}
\end{theorem}

Let $\OPT_H$ be the $(1+\eps)$-approximate solution returned by the algorithm in \Cref{thm:JS23}. Then in Step 3, we simply look at every  two consecutive vertices $\hat{u}$ and $\hat{v}$ on the same tour that correspond to two vertices $u$ and $v$ in $G$, respectively, and replace the portion of the tour between $\hat{u}$ and $\hat{v}$ by the shortest path from $u$ to $v$ in $G$. Let $\widehat{\OPT}_G$ be the obtained tour.   We now show that:

\begin{equation} \label{eq:CVPR-guarantee}
E[|\widehat{\OPT}_G| ]\leq (1+4\eps) |\OPT_G|
\end{equation}
and by scaling $\eps$, we get a $(1+\eps)$ approximation.

First, observe that since each edge belongs to at most $n$ tours, and it could appear at most $n$ times in a tour, the normalization in Step 1 introduces a total error of $|E(G)| n^2 \eps  |A| /(\alpha n^3) = \eps |\OPT_G|$ for a sufficiently large  constant $\alpha$.

We abuse notation by using $\OPT_G$ to denote an optimal solution of the CVPR (with unit/splittable/unsplittable demands) in $G$  \emph{after the normalization step}. Let $\widehat{\OPT}_H$ be the solution of  CVPR   in $H$ obtained by replacing every edge $(x,y)\in \OPT_G$ by a shortest path between $\eta(x)$ and $\eta(y)$ in $H$.   Let $\alpha_{x,y}$ be the number of times edge $(x,y)$ appears in all tours in $\OPT_G$.   We have:
\begin{equation}\label{eq:CVPR-cost}
\begin{split}
\mathbb{E}[|\widehat{\OPT}_H|]  &=  \sum_{(x,y)\in \OPT_G}\alpha_{(x,y) } \mathbb{E}[d_H(\eta(x),\eta(y))] \\
&\leq  (1+\eps)\sum_{(x,y) \in \OPT_G}\alpha_{x,y} d_G(x,y) = (1+\eps)|\OPT_G|
\end{split}
\end{equation}
Since $\OPT_H$ is a $(1+\eps)$-approximate solution of CVPR in $H$, $|\OPT_H| \leq (1+\eps)|\widehat{\OPT}_H|$. Thus, \Cref{eq:CVPR-cost} gives $\mathbb{E}[|\OPT_H|] \leq (1+\eps)^2|\OPT_G| \leq (1+3\eps)|\OPT_G|$ when $\eps \leq 1$. Then \Cref{eq:CVPR-guarantee} follows from that  $|\widehat{\OPT}_G|\leq |\OPT_H|$ and that the normalization step introduce an error of  $\eps |\OPT_G|$.

To complete the proof of \Cref{cor:CVRP}, it remains to analyze the running time. The embedding  step (Step 1) and the lifting step (Step 3) can both be done in polynomial time. The most expensive step is the second step. The running time in \Cref{cor:CVRP} follows directly from \Cref{thm:JS23} as the treewidth $k = \poly(\log n, 1/\eps)$ by \Cref{thm:main}. 

We remark that for the unsplittable/splittable versions of the CVRP, we also obtain QPTASes. However, the capacity of the vehicle  is now restricted to be $Q = n^{\log^{\Oh(1)}(n)}$ due to the restriction in \Cref{thm:JS23}. Removing this restriction is an interesting open problem for future work.

\paragraph{Capacitated $k$-median and capacitated/uncapacitated facility location.~} The algorithms for these problems use the same framework we set up above. For the normalization step, we also use a simple $\Oh(\log n)$ approximation obtained via tree embedding~\cite{FRT04}. For the dynamic programming step, we can solve both problems in time $n^{\Oh(k)}$ where $k$ is the treewidth. One could improve this running time to $(\log n/\eps)^{\Oh(k)}n^{\Oh(1)}$ by discretizing the distances to be a power of $(1+\eps)$ but this does not change the final running time of our algorithm. The final running time of the algorithm is dominated by the running time of the dynamic programming step, which is $2^{\poly(1/\eps,\log(n))}$. The approximation factor analysis is the same as the analysis of CVRP, implying \Cref{cor:FLP} and \Cref{cor:clustering}.

\section*{Acknowledgements}

Over the last few years, we have conducted various discussions on embeddings of planar graphs into low treewidth metrics.
We thank Christian Wulff-Nilsen for participating in some of them in the last few months and asking tricky questions. 
Marcin and Micha\l{} would like to acknowledge discussions on the topic with Fran\c{c}ois Dross,
when he was a postdoc in Warsaw around 2020, when in particular we have understood the limitations of the current lower bound
techniques and of the Talwar's scheme. 

\bibliographystyle{plainurl}
\bibliography{references}

\appendix
\section{Omitted proofs of tools from probability theory}\label{app:boring}

\subsection{Proof of Lemma~\ref{lem:hit-process}}
For every $X_i$, define a random variable $X_i' \in \{1,-1\}$ such that $\Prob(X_i' = 1) = \frac{1}{8}$, $X_i' \geq X_i$ almost surely, and $X_1', X_2',\ldots$ are independent.
Then, for every $i$, $\sum_{j=1}^i X_j \leq \sum_{j=1}^i X_j'$ so, for a variable $Z'$ defined analogously for $X_1',X_2',\ldots$, we have $Z' \geq Z$ almost surely.
Thus, it suffices to prove the lemma for the sequence $X_1',X_2',\ldots$ and the variable $Z'$. 
By slightly abusing the notation, we rename $X_i'$ to $X_i$ and $Z'$ to $Z$, that is, we assume $\Prob(X_i = 1) = \frac{1}{8}$ for every $i \geq 1$.

Let $Y = \sum_{j=1}^{2k} X_j$. Then the condition $k + Y \geq 0$ is equivalent to the following condition: for at least $\frac{k}{2}$ indices $i \in \{1,\ldots,2k\}$ we have $X_i = 1$. Since the expected number of such indices is $\frac{k}{4}$, by Theorem~\ref{thm:MultChernoff} applied for $\delta = 1$, we have
\[ \Prob \left(k + Y \geq 0\right) \leq \exp\left(-\frac{k}{12}\right). \]
As the event ``$k + Y \geq 0$'' contains the event ``$Z>2k$'', the proof is finished.

\subsection{Proof of Lemma~\ref{lem:memory-application}}
Fix measurable sets $S_1,S_2,\ldots,S_n \subseteq \R_{\geq 0}$ and for $1 \leq j \leq n$, let $p_j$ be the probability that a random variable with distribution $\ExpDist(1)$ has its value in $S_j$. 
To prove the lemma, it suffices to show that
\begin{equation}\label{eq:belkot:goal}
\Prob\left(Y_1 \in S_1 \wedge \ldots \wedge Y_n \in S_n\right) = \prod_{j=1}^n p_j.
\end{equation}
We will need the following computation.
Fix $1 \leq a \leq 2n$, $1 \leq j \leq n$, and $x_1,x_2,\ldots,x_{a-1} \in \R_{\geq 0}$.
\begin{align}
&\sum_{b=a+1}^{2n} \int_0^{f_{a+1}(x_1,\ldots,x_{a})} e^{-x_{a+1}} \dd x_{a+1}  \ldots \int_0^{f_{b-1}(x_1,\ldots,x_{b-2})} e^{-x_{b-1}} \dd x_{b-1} \nonumber\\
& \qquad \int_{f_{b}(x_1,\ldots,x_{b-1})}^{+\infty} e^{-x_{b}} [ x_{b} - f_{b}(x_1,\ldots,x_{b-1}) \in S_j ] \dd x_{b} \nonumber\\
& = \sum_{b=a+1}^{2n} \int_0^{f_{a+1}(x_1,\ldots,x_{a})} e^{-x_{a+1}} \dd x_{a+1}  \ldots \int_0^{f_{b-1}(x_1,\ldots,x_{b-2})} e^{-x_{b-1}} \dd x_{b-1} \nonumber\\
& \qquad e^{-f_{b}(x_1,\ldots,x_{b-1})} \int_0^{+\infty} e^{-y} [y \in S_j] \dd y \nonumber\\
& = \sum_{b=a+1}^{2n} \int_0^{f_{a+1}(x_1,\ldots,x_{a})} e^{-x_{a+1}} \dd x_{a+1}  \ldots \int_0^{f_{b-1}(x_1,\ldots,x_{b-2})} e^{-x_{b-1}} \dd x_{b-1} \nonumber\\
& \qquad e^{-f_{b}(x_1,\ldots,x_{b-1})} p_j \nonumber\\ 
& = p_j \Big( e^{-f_{a+1}(x_1,\ldots,x_{a})} + \int_0^{f_{a+1}(x_1,\ldots,x_{a})} e^{-x_{a+1}} \dd x_{a+1} \nonumber\\
&\qquad \quad \Big( e^{-f_{a+2}(x_1,\ldots,x_{a+1})} + \int_0^{f_{a+2}(x_1,\ldots,x_{a+1})} e^{-x_{a+2}} \dd x_{a+2} \nonumber\\ 
&\qquad\qquad\qquad \ldots \nonumber\\
&\qquad \quad \quad \Big( e^{-f_{2n-1}(x_1,\ldots,x_{2n-2})} + \int_0^{f_{2n-1}(x_1,\ldots,x_{2n-2})} e^{-x_{2n-1}} \dd x_{2n-1} \nonumber\\
&\qquad \quad \quad\quad e^{-f_{2n}(x_1,\ldots,x_{2n-1})} \Big)\Big)\ldots\Big) \nonumber\\
& = p_j \Big( e^{-f_{a+1}(x_1,\ldots,x_{a})} + \int_0^{f_{a+1}(x_1,\ldots,x_{a})} e^{-x_{a+1}} \dd x_{a+1} \nonumber\\
&\qquad \quad \Big( e^{-f_{a+2}(x_1,\ldots,x_{a+1})} + \int_0^{f_{a+2}(x_1,\ldots,x_{a+1})} e^{-x_{a+2}} \dd x_{a+2} \nonumber\\ 
&\qquad\qquad\qquad \ldots \nonumber\\
&\qquad \quad \quad \Big( e^{-f_{2n-1}(x_1,\ldots,x_{2n-2})} + \int_0^{f_{2n-1}(x_1,\ldots,x_{2n-2})} e^{-x_{2n-1}} \dd x_{2n-1} \Big)\Big)\ldots\Big) \nonumber\\
& = p_j \Big( e^{-f_{a+1}(x_1,\ldots,x_{a})} + \int_0^{f_{a+1}(x_1,\ldots,x_{a})} e^{-x_{a+1}} \dd x_{a+1} \nonumber\\
&\qquad \quad \Big( e^{-f_{a+2}(x_1,\ldots,x_{a+1})} + \int_0^{f_{a+2}(x_1,\ldots,x_{a+1})} e^{-x_{a+2}} \dd x_{a+2} \nonumber\\ 
&\qquad\qquad\qquad \ldots \nonumber\\
&\qquad \quad \quad \Big( e^{-f_{2n-2}(x_1,\ldots,x_{2n-3})} + \int_0^{f_{2n-2}(x_1,\ldots,x_{2n-3})} e^{-x_{2n-2}} \dd x_{2n-2} \Big)\Big)\ldots\Big) \nonumber\\
& = \ldots = p_j.\label{eq:belkot:step}
\end{align}
Here, we first used that $f_{2n} \equiv 0$ so 
\[  e^{-f_{2n}(x_1,\ldots,x_{2n-1})}  = 1. \]
Then, for $i=2n-1, 2n-2, \ldots$ we used
\[ e^{-f_i(x_1,\ldots,x_{i-1})} + \int_{0}^{f_i(x_1,\ldots,x_{i-1})} e^{-x_i} \dd x_i = e^{-f_i(x_1,\ldots,x_{i-1})} + 1 - e^{-f_i(x_1,\ldots,x_{i-1})} = 1. \]

We expand the left hand side of~\eqref{eq:belkot:goal}, conditioning on the values of $I_1, \ldots, I_n$,
   and then use repeatedly~\eqref{eq:belkot:step} for $a=i_{n-1},i_{n-2},\ldots,i_1,1$.
\begin{align*}
&\Prob\left(\forall_{j=1}^n Y_j \in S_j\right) = \sum_{1 \leq i_1 < i_2 < \ldots i_n \leq 2n} \Prob\left(\forall_{j=1}^n I_j = i_j \wedge X_{i_j} - f_{i_j}(X_1,\ldots) \in S_j \right) \\
    & \qquad = \sum_{i_1=1}^{2n} \int_0^{f_1} e^{-x_1} \dd x_1 \int_{0}^{f_2(x_1)} e^{-x_2} \dd x_2 \ldots \int_0^{f_{i_1-1}(x_1,\ldots,x_{i_1-2})} e^{-x_{i_1-1}} \dd x_{i_1-1} \\
    & \qquad \quad \qquad \int_{f_{i_1}(x_1,\ldots,x_{i_1-1})}^{+\infty} e^{-x_{i_1}} [ x_{i_1} - f_{i_1}(x_1,\ldots,x_{i_1-1}) \in S_1 ] \dd x_{i_1} \\
    &  \qquad \quad \sum_{i_2=i_1+1}^{2n} \int_0^{f_{i_1+1}(x_1,\ldots,x_{i_1})} e^{-x_{i_1+1}} \dd x_{i_1+1}  \ldots \int_0^{f_{i_2-1}(x_1,\ldots,x_{i_2-2})} e^{-x_{i_2-1}} \dd x_{i_2-1} \\
    &  \qquad \quad \qquad \int_{f_{i_2}(x_1,\ldots,x_{i_2-1})}^{+\infty} e^{-x_{i_2}} [ x_{i_2} - f_{i_2}(x_1,\ldots,x_{i_2-1}) \in S_2 ] \dd x_{i_2} \\
    &  \qquad \quad \qquad \qquad \qquad \qquad \ldots \\
    &  \qquad \quad \sum_{i_n=i_{n-1}+1}^{2n} \int_0^{f_{i_{n-1}+1}(x_1,\ldots,x_{i_{n-1}})} e^{-x_{i_{n-1}+1}} \dd x_{i_{n-1}+1}  \ldots \int_0^{f_{i_n-1}(x_1,\ldots,x_{i_n-2})} e^{-x_{i_n-1}} \dd x_{i_n-1} \\
    &  \qquad \quad \qquad \int_{f_{i_n}(x_1,\ldots,x_{i_n-1})}^{+\infty} e^{-x_{i_n}} [ x_{i_n} - f_{i_n}(x_1,\ldots,x_{i_n-1}) \in S_n ] \dd x_{i_n} \\
    & \qquad = p_n \sum_{i_1=1}^{2n} \int_0^{f_1} e^{-x_1} \dd x_1 \int_{0}^{f_2(x_1)} e^{-x_2} \dd x_2 \ldots \int_0^{f_{i_1-1}(x_1,\ldots,x_{i_1-2})} e^{-x_{i_1-1}} \dd x_{i_1-1} \\
    & \qquad \quad \qquad \int_{f_{i_1}(x_1,\ldots,x_{i_1-1})}^{+\infty} e^{-x_{i_1}} [ x_{i_1} - f_{i_1}(x_1,\ldots,x_{i_1-1}) \in S_1 ] \dd x_{i_1} \\
    &  \qquad \quad \sum_{i_2=i_1+1}^{2n} \int_0^{f_{i_1+1}(x_1,\ldots,x_{i_1})} e^{-x_{i_1+1}} \dd x_{i_1+1}  \ldots \int_0^{f_{i_2-1}(x_1,\ldots,x_{i_2-2})} e^{-x_{i_2-1}} \dd x_{i_2-1} \\
    &  \qquad \quad \qquad \int_{f_{i_2}(x_1,\ldots,x_{i_2-1})}^{+\infty} e^{-x_{i_2}} [ x_{i_2} - f_{i_2}(x_1,\ldots,x_{i_2-1}) \in S_2 ] \dd x_{i_2} \\
    &  \qquad \quad \qquad \qquad \qquad \qquad \ldots \\
    &  \qquad \quad \sum_{i_{n-1}=i_{n-2}+1}^{2n} \int_0^{f_{i_{n-2}+1}(x_1,\ldots,x_{i_{n-2}})} e^{-x_{i_{n-2}+1}} \dd x_{i_{n-2}+1}  \ldots \int_0^{f_{i_{n-1}-1}(x_1,\ldots,x_{i_{n-1}-2})} e^{-x_{i_{n-1}-1}} \dd x_{i_{n-1}-1} \\
    &  \qquad \quad \qquad \int_{f_{i_{n-1}}(x_1,\ldots,x_{i_{n-1}-1})}^{+\infty} e^{-x_{i_{n-1}}} [ x_{i_{n-1}} - f_{i_{n-1}}(x_1,\ldots,x_{i_{n-1}-1}) \in S_{n-1} ] \dd x_{i_{n-1}} \\
    & \qquad = \ldots  = \prod_{j=1}^n p_j.
\end{align*}
This finishes the proof.

\section{Omitted proofs on graphs with an excluded minor}\label{app:minor}

\subsection{Proof of Theorem~\ref{thm:structure}}

In this section we sketch how to derive Theorem~\ref{thm:structure} from known results.
We assume that the reader is familiar with the notions of vortices and of near-embeddings and the structure
theorem of Robertson and Seymour~\cite{gm16}.
In what follows, we will use the notation of Diestel et al.~\cite{DiestelKMW12}.

The crux lies in the following observation.

\begin{lemma}\label{lem:near-embed-apex-minor}
For every surface $\Sigma$ and constants $a,b$ there exists an apex graph $K$
such that every graph nearly-embeddable on $\Sigma$ with at most $a$ vortices of depth at most $b$ each
(and with no apices) does not contain the graph $K$ as a minor.
\end{lemma}
\begin{proof}
For an integer $t$, let $W_t$ be the $t \times t$ grid and let $\widehat{W}_t$ be the graph $W_t$ with an additional
vertex added and made adjacent to all vertices of the grid. 
We will prove that for some $t$, depending on $\Sigma$, $a$, and $b$ only, $\widehat{W}_t$ cannot be a minor
of a graph $G$ that is nearly-embedded as in the lemma statement.

Let $G$ be a graph with a near embedding as in the lemma statement and assume there is a minor model
of $\widehat{W}_t$ in $G$ for large $t$. 
We focus on the grid $W_t$ part of $\widehat{W}_t$ for a moment.
It follows from \cite[Lemma~22]{DiestelKMW12} that there is a large subgrid $W_{t'}$ of $W_t$
whose minor model is disjoint with vortices in the following sense: 
for every vertex of the subgrid $W_{t'}$, its branch set in the minor model of $W_{t'}$ is a subgraph
of its branch set in the minor model of $W_t$ and is vertex-disjoint with vortices of the near-embedding.
By possibly enlarging the branch set of the apex vertex of $\widehat{W}_t$, we obtain a minor model of $\widehat{W}_{t'}$
in $G$ such that the branch sets of the grid part are vertex-disjoint with vortices. 

If we now contract every vortex of $G$ into a single vertex, we obtain a graph embedded in $\Sigma$
(in the classic sense). 
The minor model of $\widehat{W}_{t'}$ in $G$ projects to a minor model of $\widehat{W}_{t'}$ in the obtained graph.
However, as the Euler genus of $\widehat{W}_s$ grows with $s$ to infinity, this implies that $t'$
is bounded by a function of~$\Sigma$, which in turn implies that $t$ is bounded as a function
of $\Sigma$, $a$, and $b$. This finishes the proof.
\end{proof}

Theorem~\ref{thm:structure} now follows from the structure theorem of Robertson and Seymour~\cite{gm16}
(with its algorithmic versions of~\cite{DemaineHK05,KawarabayashiW11}) where
in every bag we designate the apices of the near-embedding as the set $\apices(t)$. 

\subsection{Proof of Lemma~\ref{lem:add-simplicial}}

Let $\mathcal{G}$ be the class of $K$-minor-free graphs.
Let $\mathcal{G}'$ be the class of graphs that can be constructed from a member $G \in \mathcal{G}$
by adding, for some cliques in $G$, a new vertex adjacent to the clique. 
We observe that $\mathcal{G}'$ is contraction-closed.
Hence, the class $\mathcal{G}''$ defined as all subgraphs of graphs in $\mathcal{G}'$ is minor-closed.

Recall that a graph class $\mathcal{H}$ has \emph{locally bounded treewidth} if there exists a
function $f$ such that for every $H \in \mathcal{H}$, $v \in V(H)$, and integer $r \geq 0$,
the treewidth of the subgraph induced by the radius-$r$ ball around $v$ in $H$ is bounded by $f(r)$.
Eppstein~\cite{Eppstein00} showed that for minor-closed graph classes, the property
of having locally bounded treewidth is equivalent to excluding some apex graph as a minor.

Hence, $\mathcal{G}$ has locally bounded treewidth, say with a function $f$.
Then, $\mathcal{G}'$ has locally bounded treewidth with the function $f'(r) := f(r) + 1$.
Note that for every $G' \in \mathcal{G}'$, $r \geq 0$, subgraph $G''$ of $G'$,
and a vertex $v \in V(G'')$, the ball of radius $r$ centered at $v$ in $G''$
is a subgraph of the ball of radius $r$ centered at $v$ in $G'$.
Hence, $\mathcal{G}''$ has locally bounded treewidth with the same function $f'$.
Consequently, as $\mathcal{G}''$ is minor-closed, it excludes some apex graph by~\cite{Eppstein00}.
This concludes the proof.

\end{document}